\documentclass[11pt]{article}
\usepackage[letterpaper,margin=1in]{geometry}
\usepackage{setspace}
\usepackage[utf8]{inputenc}

\usepackage{subcaption}
\usepackage{physics}
\usepackage{amsmath,graphicx}
\usepackage{mathtools}
\mathtoolsset{showonlyrefs=true}
\usepackage{float}
\usepackage{amsmath}
\usepackage{amsfonts}
\usepackage{amsthm}
\usepackage{enumitem}  
\usepackage{xfrac}
\usepackage[ruled,vlined]{algorithm2e}
\usepackage{authblk}

\usepackage[backref=page]{hyperref}
\newif\ifbackrefshowonlyfirst
\backrefshowonlyfirstfalse
\makeatletter
\let\BR@direct@old@hyper@natlinkstart\hyper@natlinkstart
\renewcommand*{\hyper@natlinkstart}{\phantomsection\BR@direct@old@hyper@natlinkstart}
\let\BR@direct@oldBR@citex\BR@citex
\renewcommand*{\BR@citex}{\phantomsection\BR@direct@oldBR@citex}%

\long\def\hyper@page@BR@direct@ref#1#2#3{p. \hyperlink{#3}{#1}}

\ifx\backrefxxx\hyper@page@backref
    \let\backrefxxx\hyper@page@BR@direct@ref
    \ifbackrefshowonlyfirst
    \fi
\else
    \ifbackrefshowonlyfirst
    \fi
\fi

\RequirePackage{etoolbox}
\patchcmd{\Hy@backout}{Doc-Start}{\@currentHref}{}{\errmessage{I can't seem to patch backref}}
\makeatother

\usepackage{natbib}
\usepackage{blindtext} 


\newtheorem{theorem}{Theorem}[section]

\newtheorem{definition}{Definition}[section]

\newtheorem{lemma}[theorem]{Lemma}

\newtheorem{remark}{Remark}[section]

\newcommand{\inner}[2]{ \left\langle #1 \; , \; #2 \right\rangle }
\usepackage{interval}
\intervalconfig{soft open fences}
\usepackage{xcolor}
\usepackage{braket}
\usepackage{changes}

\DeclareMathOperator*{\argmin}{ \arg \min }
\DeclareMathOperator*{\argmax}{ \arg \max }

\newcommand{\intercal}{\mathsf{\scriptscriptstyle{T}}}

\title{A Linearly Convergent Algorithm for Petz-Augustin Means\\ with Applications to Computing Petz-Augustin Capacity and Equilibrium Prices in Fisher Markets}
\author[1]{Chun-Neng~Chu}
\author[1,2]{Wei-Fu~Tseng}
\author[1,2,3]{Yen-Huan Li}

\affil[1]{Department of Computer Science and Information Engineering,\protect\\National Taiwan University}
\affil[2]{Department of Mathematics, National Taiwan University}
\affil[3]{Center for Quantum Science and Engineering, \protect\\ National Taiwan University}

\date{}

\begin{document}

\maketitle

\begin{abstract}
We study the computation of the Petz-Augustin mean of order $\alpha \in (0,1) \cup (1,\infty)$, defined as the minimizer of a weighted sum of $n$ Petz-R\'enyi divergences of order $\alpha$ over the set of $d$-by-$d$ quantum states, where the Petz-R\'enyi divergence is a quantum generalization of the classical R\'enyi divergence.
We propose the first algorithm with a non-asymptotic convergence guarantee for solving this optimization problem. 
The iterates are guaranteed to converge to the Petz-Augustin mean at a linear rate of \( O\left( \lvert 1 - 1/\alpha \rvert^T \right) \) with respect to the Thompson metric for $\alpha\in(1/2,1)\cup(1,\infty)$, where \( T \) denotes the number of iterations.
The algorithm has an initialization time complexity of $O\left(nd^3\right)$ and a per-iteration time complexity of $O\left(nd^2 + d^3\right)$. 

Two applications follow. 
The first application concerns the Petz-Augustin capacity of order $\alpha\in(1/2,1)$, which generalizes the quantum channel capacity and characterizes the optimal error exponent in classical-quantum channel coding.
To compute this quantity, we propose the first iterative method with a non-asymptotic convergence guarantee, which uses our algorithm for computing the Petz-Augustin mean as a subroutine.
The second application concerns the equilibrium prices in Fisher markets with constant elasticity of substitution (CES) utilities.
We establish that when all buyers have CES utilities of common elasticity, the equilibrium prices are equivalent to the Petz-Augustin mean when all quantum states commute.
In this case, our proposed algorithm is a t\^{a}tonnement process \`{a} la Walras.
We further extend the algorithm to inhomogeneous Fisher markets, where buyers have CES utilities with different elasticities, and prove a faster convergence rate than existing t\^{a}tonnement-type algorithms.

\end{abstract}

\section{Introduction}
\label{sec:Intro}
We consider the problem of computing the Petz-Augustin mean $Q_{\star}$ of order $\alpha \in (0, 1) \cup (1, \infty)$, which is defined as \citep{Cheng2022}
\begin{equation}
\label{def:PetzAugMean}
Q_{\star} \in \argmin_{Q \in \mathcal{D}_d} F(Q), \quad F(Q) 
\coloneqq \sum_{j=1}^n w[j] D_\alpha \left( A_j \Vert Q \right).
\end{equation}
Here, $\mathcal{D}_d$ denotes the set of quantum states---Hermitian positive semi-definite matrices with unit trace---in $\mathbb{C}^{d \times d}$. 
The matrices 
$A_j$ also belong to $\mathcal{D}_d$, and the weights $w[j] > 0$ satisfy $\sum_{j=1}^n w[j] = 1$. 
For any $A\in\mathcal{D}_d$ and any Hermitian positive semi-definite matrix $Q$, the quantity $D_{\alpha} ( A \Vert Q )$ denotes the Petz-R\'enyi divergence of order $\alpha$ \citep{Petz1986}:
\begin{equation}
    \label{def:PetzRenyiDiv}
    D_{\alpha}(A \| Q) 
    \coloneqq  \frac{1}{\alpha - 1} \log \left(\Tr\left[A^{\alpha} Q^{1-\alpha}\right]\right),
\end{equation}
whenever 
the right-hand side
is well-defined, and $D_\alpha ( A \| Q ) \coloneqq + \infty$ otherwise. 
We assume that the matrix $\sum_{j=1}^n A_j$ is full-rank, which ensures that the Petz-Augustin mean $Q_{\star}$ exists \citep[Lemma IV.8]{Mosonyi2021} and is full-rank \citep[Lemma IV.11]{Mosonyi2021}.
If this assumption does not hold, we may project all matrices onto a lower-dimensional subspace.

When all quantum states involved in the optimization problem~\eqref{def:PetzAugMean} commute, referred to as the commuting case, the problem reduces to the optimization problem defining the classical Augustin mean $q_{\star}$
\citep{Augustin1978}: 
\begin{equation}
    \label{def:ClassicalAugMean}
    q_{\star} \in \argmin_{q \in \Delta_d} f(q), \quad f(q) 
    \coloneqq \sum_{j=1}^n w[j] \frac{1}{\alpha-1}\log\left(\sum_{i=1}^d a_j[i]^{\alpha}q[i]^{1-\alpha}\right),
\end{equation}
for some $a_j \in \Delta_d$, where $a_j [i]$ denotes the $i^{\text{th}}$ coordinate of $a_j$, and $\Delta_d$ denotes the probability simplex in $\mathbb{R}^d$.

\paragraph{Motivation.}
The Petz-Augustin mean arises in a variety of contexts.
In quantum information theory, the minimum value of the optimization problem~\eqref{def:PetzAugMean} generalizes the quantum mutual information and plays a key role in defining the Petz-Augustin capacity, which characterizes the optimal error exponent in classical-quantum channel coding \citep{Dalai2013,Dalai2014,Cheng2019,Renes2025}. 
Beyond quantum information theory, the Petz-Augustin mean, particularly in the commuting case, has applications in other domains. 
For example, as established in Section \ref{sec:Fisher}, the classical Augustin mean corresponds to the equilibrium prices in Fisher markets, characterizing the steady state of price updates \citep{walras2014}. 
It has also attracted attention in the machine learning community, where the optimization problem~\eqref{def:ClassicalAugMean} has arisen in tasks such as probabilistic model aggregation \citep{Storkey2012,Storkey2015} and multiview clustering \citep{Joshi2016}.

However, despite its importance, a significant gap remains: while the computation of the classical Augustin mean has been studied in the existing literature \citep{Augustin1978,Tsai2024,Wang2024,Kamatsuka2024}, it remains unclear whether these results extend to the non-commutative setting.
Meanwhile, recent work in quantum information theory has primarily focused on computing quantities involving quantum relative entropy \citep{Fawzi2023,He2024b,Brown2024,Frenkel2023,Huang2024,Jencova2024,Kossmann2024a,Kossmann2024b,Hayashi2024,He2024a}, while relatively less attention has been given to the computation of quantities involving quantum R\'{e}nyi divergences \citep{You2022,He2025,Burri2025}.
In particular, to the best of our knowledge, no existing algorithm provides a non-asymptotic convergence guarantee for computing the Petz-Augustin mean.

Designing such an algorithm involves several challenges.
To begin with, the Petz-Augustin mean does not admit a closed-form expression, which necessitates an iterative method for solving the optimization problem~\eqref{def:PetzAugMean}.
Furthermore, although the optimization problem~\eqref{def:PetzAugMean} is convex for orders $\alpha \in (0, 1) \cup (1, 2]$ \citep{Mosonyi2017}, the gradient and Hessian of the objective function $F$ are unbounded \citep[Propositions 3.1 and 3.2]{You2022}, 
violating standard assumptions in the convex optimization literature. 
As a result, standard first-order optimization algorithms, such as gradient descent and mirror descent, and their theoretical guarantees do not directly apply.
Although second-order methods \citep{Nesterov2018a}, such as Newton's method, may be applicable, their per-iteration time complexity grows rapidly with the dimension $d$ of the quantum states.
As a result, these methods do not scale well with the number of qubits, making them prohibitive for quantum information applications.

\paragraph{Our Contributions.}
\begin{itemize}
    \item In Section \ref{sec:RelationWithCP}, we identify a similarity between the optimization problem~\eqref{def:PetzAugMean} and the one defining the $\ell_p$-Lewis weights.
    Inspired by this similarity and an existing algorithm for computing the $\ell_p$-Lewis weights \citep{Cohen2015}, we propose the following iteration rule for computing the Petz-Augustin mean \eqref{def:PetzAugMean}, formally defined in \eqref{alg:ProposedIter}:
    \begin{equation}
        \label{eq:IntroAugMeanIter}
        Q_{t+1}
        =\left(\sum_{j=1}^n w[j]\frac{A_j^{\alpha}}{\Tr\left[A_j^{\alpha}Q_t^{1-\alpha}\right]}\right)^{1/\alpha},
    \end{equation}
    where $Q_t$ denotes the iterate at step $t$.
    \item In Section \ref{sec:alg}, we prove that the proposed iteration rule \eqref{eq:IntroAugMeanIter} converges at a linear rate of $O\left(\abs{1-1/\alpha}^T\right)$ for all $\alpha \in (1/2,1)\cup(1,\infty)$, making it the first algorithm with a non-asymptotic convergence guarantee for computing the Petz-Augustin mean.
    \item In Section \ref{sec:Capacity}, we apply the iteration rule \eqref{eq:IntroAugMeanIter} to develop an iterative method for computing the Petz-Augustin capacity of order $\alpha\in(1/2,1)$, which converges at a rate of $O\left(\log(n)/T\right)$. Notably, this is the first algorithm for computing the Petz-Augustin capacity with a non-asymptotic convergence guarantee.
    \item In Section \ref{sec:Fisher}, we establish that the classical Augustin mean \( q_{\star} \) is equivalent to the equilibrium prices in Fisher markets with constant elasticity of substitution (CES) utilities with a common CES parameter $\rho=1-1/\alpha$, where $\alpha$ is called the \emph{elasticity} in the literature \citep{Cole2008}.
    We show that the iteration rule \eqref{eq:IntroAugMeanIter} is a \emph{t\^{a}tonnement} process \`{a} la \citet{walras2014}. 
    We then consider a more general setting in which each \( j^{\text{th}} \) buyer has a CES utility with heterogeneous parameters \( \rho_j \in (0,1) \), which corresponds to the weak gross substitutes (WGS) regime.  
    We adapt the iteration rule \eqref{eq:IntroAugMeanIter} to handle asynchronous price updates, assuming that each \( i^{\text{th}} \) seller knows only an upper bound \( \hat{\rho}_i \in \left[\max_j \rho_j, 1\right) \).  
    In this case, the prices are guaranteed to converge to the equilibrium at a rate of \( O\left(\hat{\rho}^{T}\right) \), where \( \hat{\rho} = \max_i \hat{\rho}_i \) and $T$ denotes the number of epochs, with each seller updating their price at least once per epoch.
    This yields a faster convergence rate, in terms of the distance between the prices and the equilibrium, than all existing t\^{a}tonnement-type algorithms.
    \item 
    In addition, we point out that the convergence analysis of \citet{Dvijotham2022} for computing equilibrium prices in Fisher markets via \emph{best response dynamics}, a price update mechanism different from the t\^{a}tonnement process, is incorrect. 
    Consequently, we do not use their claimed convergence guarantee as a benchmark for our iteration rule.
    We provide a counterexample to their claimed convergence guarantee in Appendix~\ref{appendix:DvijothamCounterexample}.
\end{itemize}

\paragraph{Key Ideas.}
\begin{itemize}
    \item To analyze the iteration rule \eqref{eq:IntroAugMeanIter} for computing the Petz-Augustin mean, 
    we
    prove that this rule is contractive with respect to the Thompson metric and that its fixed point 
    coincides with the Petz-Augustin mean.
    \item 
    To develop an algorithm with a non-asymptotic convergence guarantee for computing the Petz-Augustin capacity, formulated as a convex optimization problem $\min_{w} g(w)$, we show that the gradient $\nabla g(w)$, although lacking a closed-form expression, can be computed via the iteration rule \eqref{eq:IntroAugMeanIter}. 
    This enables the use of gradient-based methods.  
    To identify a suitable gradient-based method, 
    we show that \( g(w) \) is 1-smooth relative to the negative Shannon entropy \citep{Lu2018}.
    This property is sufficient to ensure the convergence of entropic mirror descent, which is our method of choice, at a rate of \( O\left(\log(n)/T\right) \) \citep{Lu2018,He2024a,Birnbaum2011}.
    \item To adapt the iteration rule \eqref{eq:IntroAugMeanIter} for computing equilibrium prices in Fisher markets, 
    we 
    first express it, when applied to Fisher markets with a common CES parameter \( \rho \) in buyers' utilities, as follows:
    \begin{equation}
        \label{def:ProposedClassicalIterMultiGrad}
        p_{t+1}[i]
        = p_t[i] \cdot x(p_t)[i]^{1 - \rho},
    \end{equation}
    where \( p_t[i] \) and \( x(p_t)[i] \) denote the price and total demand for the \( i^{\text{th}} \) good at round \( t \), respectively.  
    Then, we interpret the exponent \( 1 - \rho \) in \eqref{def:ProposedClassicalIterMultiGrad} as a step size and adjust it appropriately on a per-coordinate basis.
\end{itemize}


\section{Related Work}
\subsection{Augustin Mean}
\paragraph{Commuting Case.}
For computing the classical Augustin mean \eqref{def:ClassicalAugMean}, the commuting case of \eqref{def:PetzAugMean}, \citet{Augustin1978} proposed a fixed-point iteration, which we 
refer to as the Augustin iteration.
It is proved to converge asymptotically for $\alpha\in(0,1)$ by 
\citet{Karakos2008} and 
\citet{Nakiboglu2019}, and at a linear rate with respect to the Hilbert projective metric for $\alpha\in(1/2,1)\cup(1,3/2)$ by \citet{Tsai2024}.

Our proposed method can be viewed as a generalization of the Augustin iteration.
In particular, the Augustin iteration for solving the optimization problem~\eqref{def:ClassicalAugMean} can be written as
\begin{equation}
    \label{eq:AugIter}
    q_{t+1}[i]
    = q_t[i]\left(-\nabla f(q_t)[i]\right),\quad \forall i,
\end{equation}
where $q_t$ denotes the $t^{\text{th}}$ iterate of the algorithm, and $q_t[i]$ denotes its $i^{\text{th}}$ coordinate.
On the other hand, the proposed iteration rule \eqref{eq:IntroAugMeanIter} can be expressed as 
\begin{align}
     q_{t+1}[i]
    = q_t[i]\left(-\nabla f(q_t)[i]\right)^{1/\alpha},\quad \forall i.
\end{align} 
Thus, our method can be viewed as a generalization of the Augustin iteration with an additional parameter $1/\alpha$ in the exponent, which is analogous to the step size in first-order optimization methods. 

For algorithms different from the Augustin iteration, an alternating minimization method \citep{Kamatsuka2024} converges at a rate of $O(1/T)$ for $\alpha\in(1,\infty)$ \citep{Tsai2024}, where $T$ denotes the number of iterations. 
Riemannian gradient descent with respect to the Poincar{\'e} metric 
also converges at a rate of $O(1/T)$ for all $\alpha\in(0,1)\cup(1,\infty)$ \citep{Wang2024}.

\paragraph{Quantum Case.}
While the computation of the classical Augustin mean \eqref{def:ClassicalAugMean} has been well studied, efficient algorithms for computing the Petz-Augustin mean \eqref{def:PetzAugMean} are still lacking.
To the best of our knowledge, there is currently no 
first-order method
that guarantees non-asymptotic convergence for any $\alpha\in(0,1)\cup(1,\infty)$.

Since the objective function $F$ in \eqref{def:PetzAugMean} has a locally bounded gradient, entropic mirror descent with Armijo line search \citep{Li2019a} or with the Polyak step size \citep{You2022} is applicable for $\alpha\in(0,1)\cup\left(1,2\right]$. 
However, these two algorithms only guarantee asymptotic convergence.

Our proposed iteration rule \eqref{eq:IntroAugMeanIter} is inspired by an algorithm proposed by \citet{Cohen2015} for computing the $\ell_p$-Lewis weights, which has been proven to converge linearly with respect to the Thompson metric.
The $\ell_p$-Lewis weights, along with their variants, have applications in $\ell_1$-regression \citep{Durfee2018,Parulekar2021}, log-concave sampling \citep{Kook2024,Jiang2024}, and linear programming \citep{Lee2020}.
Despite the conceptual connection, our results do not immediately follow from those of \citet{Cohen2015}, as discussed in Section \ref{sec:RelationWithCP}.

There is an overlap between our proposed iteration rule \eqref{eq:IntroAugMeanIter} and the one proposed by \citet{Cheng2024}.
In particular, \citet{Cheng2024} consider the following optimization problem, which they show to be equivalent to \eqref{def:PetzAugMean}:
\begin{equation}
    \label{def:AugDual}
    v_{\star}\in\argmax_{v\in\mathbb{R}^n}H(v),
    \quad H(v)
    \coloneqq \frac{1-\alpha}{\alpha}\sum_{j=1}^n w[j]v[j]-\log\left(\Tr\left[\left(\sum_{j=1}^n w[j]\exp((1-\alpha)v[j])A_j^{\alpha}\right)^{1/\alpha}\right]\right),
\end{equation}
where the matrices $A_j$, weights $w[j]$, and order $\alpha$ are the same as in \eqref{def:PetzAugMean}, and $v[j]$ denotes the $j^{\text{th}}$ coordinate of $v$.
They propose updating their iterates $v_t$ according to
\begin{equation}
    \label{alg:ChengDualIter}
    v_{t+1}[j]
    =D_{\alpha}\left(A_j \| \mu(v_t)\right),
    \quad \mu(v)
    \coloneqq \left(\sum_{j=1}^n w[j]\exp((1-\alpha)v[j])A_j^{\alpha}\right)^{1/\alpha}.
\end{equation}
By initializing our iteration rule \eqref{eq:IntroAugMeanIter} with $Q_1=\mu(v_1)$, the resulting iterates $Q_t$ coincide with $\mu(v_t)$ for all $t\in\mathbb{N}$.
Despite the correspondence between the algorithms, \citet{Cheng2024} only establish asymptotic convergence for $\alpha\in(1,\infty)$.
In contrast, our derivation and analysis adopt a different perspective, and we prove that the algorithm converges at a rate of $O\left(\left|1-1/\alpha\right|^T\right)$ for $\alpha\in(1/2,1)\cup(1,\infty)$, thereby not only extending the range of $\alpha$ but also providing a non-asymptotic convergence guarantee. 

\subsection{Petz-Augustin Capacity}
The Petz-Augustin capacity generalizes the quantum channel capacity \citep{Mosonyi2017}. However, unlike the quantum channel capacity, which can be rigorously computed using the quantum Blahut-Arimoto algorithm \citep{Hayashi2024,He2024a,Li2019b,Nagaoka1998,Ramakrishnan2021}, there is no algorithm known to compute the Petz-Augustin capacity with any convergence guarantee.
Although several alternating minimization methods are applicable when all quantum states commute \citep{Arimoto1976,Kamatsuka2024b,Kamatsuka2024c,Jitsumatsu2020}, it remains unclear how to extend these approaches to the non-commutative setting.
To this end, inspired by \citet{He2024a}, 
who interpreted the quantum Blahut-Arimoto algorithm as entropic mirror descent, 
we show in Section \ref{sec:Capacity} that a similar approach is applicable to computing the Petz-Augustin capacity, using our proposed iteration rule \eqref{eq:IntroAugMeanIter} for computing the Petz-Augustin mean as a subroutine.
This yields the first algorithm for computing the Petz-Augustin capacity with a non-asymptotic convergence guarantee.

\subsection{Equilibrium Prices in Fisher Markets}
Equilibrium prices in Fisher markets can be computed using various algorithms, some of which also admit interpretations as market dynamics, such as t\^{a}tonnement processes \citep{Codenotti2005,Cole2008,Cole2010,Goktas2023,Nan2025,Cheung2020,Shikhman2018,Cheung2012,Bei2019,Cheung2018ESA} and \emph{proportional response dynamics} \citep{Wu2007,Birnbaum2011,Zhang2011,Cheung2018,Kolumbus2023,Branzei2021,Cheung2025}.
Notably, due to the correspondence between the classical Augustin mean and the equilibrium prices in Fisher markets, clarified in Section \ref{sec:Fisher}, it can be shown that the Augustin iteration \citep{Augustin1978}, originally proposed in the context of information theory, is equivalent to the algorithm of \citet{Cheung2018} when applied to Fisher markets with CES utilities with a common parameter $\rho$ (see Appendix \ref{appendix:EqAugTatonnement}). 
In particular, the Augustin iteration can be interpreted as both a t\^{a}tonnement process and a proportional response dynamic and was proved by \citet{Cheung2018} to converge linearly for $\rho\in(-\infty,0)$, which corresponds to the regime of complements.
Recently, \citet{Tsai2024} show that the Augustin iteration converges linearly for  $\rho\in(-1,1/3)$, thereby complementing the results of \citet{Cheung2018}.

Our proposed iteration rule admits a t\^{a}tonnement interpretation.
It applies to the case where $\rho\in(0,1)$, which corresponds to the WGS regime.
As shown in Section \ref{subsec:TatonnementCmp}, the proposed iteration rule converges faster than all existing t\^{a}tonnement-type algorithms for the WGS regime in terms of the distance between iterates and the equilibrium.


Beyond t\^{a}tonnement processes and proportional response dynamics, \citet{Dvijotham2022} proposed a different approach for computing equilibrium prices in Fisher markets with WGS CES utilities. 
Their approach can be interpreted as a best response dynamic \citep{Milgrom1991,Kreps1990}, which differs from a t\^{a}tonnement process in that a seller in a best response dynamic requires knowledge of the prices of all goods and must make assumptions about other sellers' behavior in order to determine its “best” price update.
In contrast, in a t\^{a}tonnement process, each seller updates the price of its good based only on the demand for that good and its previous price, without making any assumptions about other sellers' strategies.

While, as discussed in Section~\ref{sec:Intro}, the convergence analysis of \citet{Dvijotham2022} is incorrect in general, their technique for handling asynchronous price updates remains valid, and we leverage this technique in our analysis.

\section{Preliminaries}
\label{sec:problem}
\subsection{Notations}
We denote the sets of vectors in $\mathbb{R}^d$ with nonnegative entries and
strictly 
positive entries by $\mathbb{R}_{+}^d$ and $\mathbb{R}_{++}^d$, respectively.
We denote by $\odot$ the coordinate-wise product.
We denote the set of all Hermitian matrices in $\mathbb{C}^{d\times d}$ by $\mathcal{B}_d$.
We denote by $\mathbf{1}_d$ the all-ones vector in $\mathbb{R}^d$.
We denote by $I_d$ the identity matrix in $\mathbb{C}^{d\times d}$.


We define $A[i]$ to be the $i^{\text{th}}$ row of a matrix $A$ and $v[i]$ to be the $i^{\text{th}}$ coordinate of a vector $v$. 
For any vector $v$, we denote by $\mathrm{Diag}(v)$ the diagonal matrix whose $i$-th diagonal element is $v[i]$. 

For any function \( f: \mathbb{R} \to \mathbb{R} \) and 
vector \( v \in \mathbb{R}^d \), we define \( f(v) \) 
as 
the 
\( d \)-dimensional vector where \( f(v)[i] = f(v[i]) \). 
Similarly, for such a function \( f \) and \( Q \in \mathcal{B}_d \), we define \( f(Q) \) as \( \sum_{i=1}^d f(\lambda_i)  P_i \), where \( Q = \sum_{i=1}^d \lambda_i   P_i \) is the eigendecomposition of \( Q \).


For any \( Q_1, Q_2 \in \mathcal{B}_d \), we write \( Q_1 \leq Q_2 \) if and only if \( Q_2 - Q_1 \) is positive semi-definite. Similarly, we write \( Q_1 < Q_2 \) if and only if \( Q_2 - Q_1 \) is positive definite.
For any $q_1,q_2\in\mathbb{R}^{d}$, we write $q_1\leq q_2$ if and only if $\mathrm{Diag}(q_1)\leq\mathrm{Diag}(q_2)$.
Similarly, we write $q_1<q_2$ if and only if $\mathrm{Diag}(q_1)<\mathrm{Diag}(q_2)$.
We define $\mathcal{B}_{d,+}$ and $\mathcal{B}_{d,++}$ as the nonnegative cone and the positive cone in \( \mathcal{B}_d \), respectively. 
Specifically,
\begin{equation} 
\mathcal{B}_{d,+} \coloneqq \Set{ Q \in \mathcal{B}_d \mid Q \geq 0},
\end{equation} 
and
\begin{equation} 
\mathcal{B}_{d,++} \coloneqq \Set{ Q \in \mathcal{B}_d \mid Q > 0}.
\end{equation}
For any $Q\in\mathcal{B}_{d,+}$, we denote its $i^{\text{th}}$ eigenvalue, ordered in decreasing order, by $\lambda_i(Q)$. 

For any convex differentiable function $h:\mathbb{R}^n\to\mathbb{R}$, we define the Bregman divergence of $h$ as
\begin{equation}
    B_h(v\Vert u)
    \coloneqq h(v)-h(u)-\inner{\nabla h(u)}{v-u},\quad\forall u\in\mathrm{dom}(\nabla h),v\in\mathrm{dom}(h).
\end{equation}


\subsection{Thompson Metric}
The Thompson metric is a useful tool to study the behavior of dynamical systems \citep{Krause2015,Lemmens2012,Nussbaum1988,Thompson1963}. 
It 
can be defined on the interiors of any normal cone in real Banach spaces.
Here, we are only interested in the following definition of the Thompson metric, specialized for $\mathcal{B}_{d,++}$ and $\mathbb{R}_{++}^d$.
\begin{definition}[{\citep{Thompson1963}}]
    \label{eq:Thompson}
    The Thompson metric between any $U,V\in\mathcal{B}_{d,++}$ 
    is given by 
    \begin{equation}
        d_{\mathrm{T}}(V,U)\coloneqq \inf 
        \Set{ r\geq0 \mid \exp(-r) V\leq U \leq \exp(r)V }.
    \end{equation}
    Similarly, the Thompson metric between any $u,v\in\mathbb{R}_{++}^d$ is given by 
    \begin{equation}
        d_{\mathrm{T}}(v,u)\coloneqq \inf 
        \Set{ r\geq0 \mid \exp(-r) v\leq u \leq \exp(r)v }=\max_{i}\log\left(\max\Set{\frac{v[i]}{u[i]},\frac{u[i]}{v[i]}}\right),
    \end{equation}
    where, for convenience, we use the same notation $d_{\mathrm{T}}$ to denote the Thompson metric on both $\mathcal{B}_{d,++}$ and $\mathbb{R}_{++}^d$.
\end{definition}


We
will 
analyze the convergence of our proposed iteration rule by establishing its contractive property with respect to the Thompson metric (see Lemma \ref{ineq:Contract}, Theorem \ref{thm:PetzConv}, and Theorem \ref{thm:FisherConv}). 
These analyses rely on the following lemmas concerning the Thompson metric, specialized for $\mathcal{B}_{d,++}$ and $\mathbb{R}_{++}^d$.

\begin{lemma}[{\citep[Lemma 3]{Thompson1963}}]
    \label{eq:ThompIsMetric}
    The Thompson metric is a well-defined metric on both $\mathcal{B}_{d,++}$ and $\mathbb{R}_{++}^d$, each of which is complete with respect to this metric.
\end{lemma}
\begin{lemma}[{\citep[Proposition 1.5]{Nussbaum1988}}]
    \label{ineq:NegCurveProp}
    For any $U,V\in\mathcal{B}_{d,++}$ and $u,v\in\mathbb{R}_{++}^d$, 
    we have
    \begin{align}
        d_{\mathrm{T}}(U^r,V^r)\leq |r|d_{\mathrm{T}}(U,V),\quad\forall r\in[-1,1],
    \end{align}
    and
    \begin{align}
        d_{\mathrm{T}}(u^r,v^r)
        = |r|d_{\mathrm{T}}(u,v),\quad\forall r\in\mathbb{R},
    \end{align}
    respectively.
\end{lemma}

\begin{lemma}
    \label{ineq:ThompsonLogHomo}
    For any $U,V\in\mathcal{B}_{d,++}$ or $U,V\in\mathbb{R}_{++}^d$, and any $r>0$, we have
    \begin{align}
        d_{\mathrm{T}}\left(U,rV\right)
        \leq d_{\mathrm{T}}\left(U,V\right)+\abs{\log(r)}.
    \end{align}
\end{lemma}

\begin{proof}
    By Definition \ref{eq:Thompson}, we write
    \begin{align}
        \exp\left(-d_{\mathrm{T}}\left(V,U\right)\right)V
        \leq U
        \leq \exp\left(d_{\mathrm{T}}\left(V,U\right)\right)V.
    \end{align}
    It follows that
    \begin{align}
        \exp\left(-d_{\mathrm{T}}\left(V,U\right)-\abs{\log(r)}\right)U
        &\leq\exp\left(-d_{\mathrm{T}}\left(V,U\right)+\log(r)\right)U\nonumber\\
        &\leq rV\nonumber\\
        &\leq \exp\left(d_{\mathrm{T}}\left(V,U\right)+\log(r)\right)U\nonumber\\
        &\leq \exp\left(d_{\mathrm{T}}\left(V,U\right)+\abs{\log(r)}\right)U.\nonumber
    \end{align}
    By Definition \ref{eq:Thompson}, the inequalities above imply that 
    \begin{align}
        d_{\mathrm{T}}\left(U,rV\right)
        \leq d_{\mathrm{T}}\left(V,U\right)+\abs{\log(r)}.
    \end{align}
    This concludes the proof.
\end{proof}
 
\section{Connection with $\ell_p$-Lewis Weights}
\label{sec:RelationWithCP}
Our proposed iteration rule \eqref{eq:IntroAugMeanIter} is inspired by the algorithm of \citet{Cohen2015} for computing the $\ell_p$-Lewis weights.
In Section \ref{subsec:LewisIntro}, we introduce the problem formulation for computing the $\ell_p$-Lewis weights and review the algorithm proposed by \citet{Cohen2015}.
In Section \ref{subsec:LewisChallenge}, we discuss the challenges involved in generalizing this approach to compute the Petz-Augustin mean.

\subsection{$\ell_p$-Lewis Weights and the Algorithm of \citet{Cohen2015}}
\label{subsec:LewisIntro}
The $\ell_p$-Lewis weights are given by \citep{Lee2020}: 
\begin{equation}
    \label{eq:LewisDef}
    u_{\star}\in\argmin_{u\in m \Delta_{d}}\frac{-1}{1-\frac{2}{p}}\log\det\left(M^{\intercal}\mathrm{Diag}\left(u^{1-2/p}\right)M\right),
\end{equation}
for $p\in(0,2)\cup(2,\infty)$, where $M$ is a $d$-by-$m$ real matrix and $M^\intercal$ denotes the transpose of $M$.
\citet{Cohen2015} proposed the following iteration rule: 
\begin{equation}
    \label{eq:CPIterDef}
    u_{t+1}=T_{\mathrm{CP}}\left(u_t\right),
\end{equation}
where
\begin{align}
    T_{\mathrm{CP}}\left(u_t\right)[i]\coloneqq\left(M[i]^{\intercal}\left(M^{\intercal}\mathrm{Diag}\left(u_t^{1-2/p}\right)M\right)^{-1}M[i]\right)^{p/2},
\end{align}
for each $i\in\set{1,2,\dots,d}$.
They proved that 
the 
iterates 
$u_t$
converge linearly to 
$u_{\star}$
with respect to the Thompson metric for $p\in(0,4)$.

To demonstrate the similarity between their proposed iteration rule and ours \eqref{eq:IntroAugMeanIter}, we note that when $m=1$, the matrix $M$ reduces to a vector in $\mathbb{R}^d$, and the $\ell_p$-Lewis weights can be written as 
\begin{equation}
    \label{eq:ReWrittenLewis}
    u_{\star}\in\argmin_{u\in \Delta_{d}}\frac{-1}{1-\frac{2}{p}}\log\left(\Tr\left[\mathrm{Diag}\left(a\right)^{2/p}\mathrm{Diag}\left(u\right)^{1-2/p}\right]\right),
\end{equation}
where $a=\abs{M}^p$.
Let $\alpha=2/p$. 
Then, the problem corresponds to a special case of the optimization problem~\eqref{def:PetzAugMean} where all 
matrices commute.
In this case, 
the iteration rule proposed by \citet{Cohen2015}
can be expressed as
\begin{equation}
    \label{eq:RawCohenIteration}
    \mathrm{Diag}\left(u_{t+1}\right)=\left(\frac{\mathrm{Diag}\left(a\right)^{\alpha}}{\Tr\left[\mathrm{Diag}\left(a\right)^{\alpha}\mathrm{Diag}\left(u_t\right)^{1-\alpha}\right]}\right)^{1/\alpha} , 
\end{equation}
where we deliberately write the iterates as matrices to illustrate the similarity with our iteration rule \eqref{eq:IntroAugMeanIter}. 

\subsection{Challenges in Generalizing the Algorithm of \citet{Cohen2015}}
\label{subsec:LewisChallenge}
Despite the similarity in algorithms, our results cannot be trivially derived from the work of \citet{Cohen2015}.
Specifically, there are three main challenges in adapting their proposed algorithm to compute the Petz-Augustin mean:
\begin{itemize}
    \item The optimization problem~\eqref{def:PetzAugMean} includes an additional weighted-sum term, $\sum_{j=1}^n w[j](\cdot)$.
    Consequently, even when all the matrices commute, the computation of the Petz–Augustin mean cannot be reduced to that of the $\ell_p$-Lewis weights.
    \item \citet{Cohen2015} established the contractive property of their proposed algorithm with respect to the Thompson metric by proving
    \begin{align}
        \label{ineq:CohenIterLewisContractive}
        \log\left(\max\Set{\frac{T_{\mathrm{CP}}(u)[i]}{T_{\mathrm{CP}}(v)[i]}, \frac{T_{\mathrm{CP}}(v)[i]}{T_{\mathrm{CP}}(u)[i]}}\right)\leq\left|1-\frac{p}{2}\right|\max_{1\leq j\leq d}\log\left(\max\Set{\frac{u[j]}{v[j]}, \frac{v[j]}{u[j]}}\right),\quad
    \end{align}
    for each $1 \leq i \leq d$.
    This can be viewed as comparing the eigenvalues of the commuting matrices $\mathrm{Diag}\left(T_{\mathrm{CP}}(u)\right)$ and $\mathrm{Diag}\left(T_{\mathrm{CP}}(v)\right)$.
    However, in our case, the iterates $Q_t$ and the Petz-Augustin mean $Q_{\star}$ may not commute, making it infeasible to compare their eigenvalues within matched eigenspaces.
    \item A natural extension of the contractive property \eqref{ineq:CohenIterLewisContractive} to the non-commutative setting fails.
    In particular, our proposed iteration rule \eqref{eq:IntroAugMeanIter} involves repeatedly applying the operator
    \begin{equation}
        \label{eq:WrongIterMap}
        \hat{T}_{F}(Q)
        \coloneqq \left(\sum_{j=1}^n w[j]\frac{A_j^{\alpha}}{\Tr\left[A_j^{\alpha}Q^{1-\alpha}\right]}\right)^{1/\alpha}
    \end{equation}
    to the iterates.
    Then, it is natural to write the contractive property \eqref{ineq:CohenIterLewisContractive} as
    \begin{align}
        \label{ineq:CohenIterLewisContractiveThomp}
        d_{\mathrm{T}}\left(T_{\mathrm{CP}}(v),T_{\mathrm{CP}}(u)\right)
        &=\max_{1\leq i\leq d}\log\left(\max\Set{\frac{T_{\mathrm{CP}}(u)[i]}{T_{\mathrm{CP}}(v)[i]}, \frac{T_{\mathrm{CP}}(v)[i]}{T_{\mathrm{CP}}(u)[i]}}\right)\nonumber\\
        &\leq\left|1-\frac{p}{2}\right|\log\left(\max_{1\leq j\leq d}\max\Set{\frac{u[j]}{v[j]}, \frac{v[j]}{u[j]}}\right)\nonumber\\
        &=\left|1-\frac{p}{2}\right|d_{\mathrm{T}}(v,u),
    \end{align}
    and to hypothesize that the following contractive property, which extends \eqref{ineq:CohenIterLewisContractiveThomp} with $\alpha=2/p$, would hold:
    \begin{equation}
        \label{ineq:WrongContraction}
        d_{\mathrm{T}}\left(\hat{T}_F(V),\hat{T}_F(U)\right)
        \leq \abs{\frac{\alpha-1}{\alpha}}d_{\mathrm{T}}\left(V,U\right).
    \end{equation}
    However, this fails. A counterexample can be constructed by setting $n=1$, $\alpha=3$, 
    \begin{equation}
        \label{eq:WrongContractionCounterEx}
        A_1
        =\frac{1}{20.6364}\begin{bmatrix}
        19.5364 & 4.42\\
        4.42 & 1.1
        \end{bmatrix},
        \quad U
        =\begin{bmatrix}
        2/3 & 1/3\\
        1/3 & 1/3
        \end{bmatrix},
        \quad \text{and}\quad V
        =\begin{bmatrix}
        1/2.1 & 1/2.1\\
        1/2.1 & 1.1/2.1
        \end{bmatrix},
    \end{equation}
    which yields
    \begin{equation}
         1.4366
         \approx d_{\mathrm{T}}\left(\hat{T}_F(V),\hat{T}_F(U)\right)
        > \abs{\frac{\alpha-1}{\alpha}}d_{\mathrm{T}}\left(V,U\right)
        \approx 1.3668,
    \end{equation}
    violating the hypothesized contractive property \eqref{ineq:WrongContraction}.
\end{itemize}

For the first challenge, we identify an appropriate generalization of the algorithm proposed by \citet{Cohen2015}, leading to a new algorithm \eqref{eq:IntroAugMeanIter} for computing the Petz-Augustin mean.

For the second and the third challenges, we leverage the properties of the Thompson metric specialized for $\mathcal{B}_{d,++}$ (Lemma \ref{eq:ThompIsMetric}, Lemma \ref{ineq:NegCurveProp}, and Lemma \ref{ineq:ThompsonLogHomo}).
Specifically, we observe that the failure of the contractive property \eqref{ineq:WrongContraction} stems from the fact that matrix powers $Q\mapsto Q^{1-\alpha}$ and $(\cdot)\mapsto(\cdot)^{1/\alpha}$  in \eqref{eq:WrongIterMap} are not, in general, order-preserving or order-reversing when $\alpha\in(2,\infty)$ and $\alpha\in(0,1)$, respectively \citep[Theorem V.2.10]{Bhatia1997}.
Therefore, even though Definition \ref{eq:Thompson} gives
\begin{equation}
    \exp(-d_{\mathrm{T}}(V,U))V
    \leq U
    \leq\exp(d_{\mathrm{T}}(V,U))V,
\end{equation}
the inequalities
\begin{equation}
    \exp((1-\alpha)d_{\mathrm{T}}(V,U))\Tr\left[A_1^{\alpha} V^{1-\alpha}\right]
    \leq \Tr\left[A_1^{\alpha} U^{1-\alpha}\right]
    \leq \exp((\alpha-1)d_{\mathrm{T}}(V,U))\Tr\left[A_1^{\alpha} V^{1-\alpha}\right]
\end{equation}
fail for the choices of $\alpha,A_1,U,V$ in the counterexample \eqref{eq:WrongContractionCounterEx}.
This failure leads to the breakdown of the following inequalities:
\begin{equation}
    \exp((1/\alpha-1)d_{\mathrm{T}}(V,U))\hat{T}_F(V)
    \leq \hat{T}_F(U)
    \leq \exp((1-1/\alpha)d_{\mathrm{T}}(V,U))\hat{T}_F(V),
\end{equation}
which, by Definition \ref{eq:Thompson}, implies that the operator $\hat{T}_F$ is not contractive with a ratio of $\abs{1-1/\alpha}$ for all $\alpha\in(1/2,1)\cup(1,\infty)$.
Our key idea is to define the corrected operator $T_{F}$ (formally defined in \eqref{eq:ProposedOperator}) as follows:
\begin{equation}
    \label{eq:CorrectIterMap}
    T_{F}(\tilde{Q})
    \coloneqq \left(\sum_{j=1}^n w[j]\frac{A_j^{\alpha}}{\Tr\left[A_j^{\alpha}\tilde{Q}\right]}\right)^{(1-\alpha)/\alpha},
\end{equation}
so that all matrix powers involved are order-preserving or order-reversing for $\alpha\in(1/2,1)\cup(1,\infty)$.
In Lemma \ref{ineq:Contract}, we prove that the operator $T_F$ is indeed contractive with a ratio of $\abs{1-1/\alpha}$.
Moreover, the equivalence of our proposed iteration rule:
\begin{equation}
    Q_{t+1}
    \coloneqq \hat{T}_F\left(Q_t\right)
    =T_F\left(Q_{t}^{1-\alpha}\right)^{1/(1-\alpha)}
\end{equation}
leads to our main result, Theorem \ref{thm:PetzConv}, which bounds the Thompson metric between $Q_{\star}^{1-\alpha}$ and $Q_{T+1}^{1-\alpha}$, rather than between $Q_{\star}$ and $Q_{T+1}$.


\section{Computing Petz-Augustin Mean}
\label{sec:alg}
\subsection{Iteration Rule}
\label{subsec:CohenIteration}
Given the objective function $F$ in the optimization problem~\eqref{def:PetzAugMean}, we define the operator
\begin{equation}
    \label{eq:ProposedOperator}
    T_{F} \colon \mathcal{B}_{d,++}\to\mathcal{B}_{d,++}:U\mapsto\left(\sum_{j=1}^n w[j]\frac{A_j^{\alpha}}{\Tr\left[A_j^{\alpha}U\right]}\right)^{(1-\alpha)/\alpha}.
\end{equation}
Below, we present our proposed iteration rule for solving the optimization problem~\eqref{def:PetzAugMean}, which is equivalent to \eqref{eq:IntroAugMeanIter} but reformulated to facilitate its convergence analysis.
\begin{itemize}
\item Let $Q_1\in\mathrm{relint}\left(\mathcal{D}_d\right)$.
\item For every $t \in \mathbb{N}$, compute 
\begin{equation}
    \label{alg:ProposedIter}
    Q_{t+1} = T_{F} \left(Q_{t}^{1-\alpha} \right)^{1/(1-\alpha)},
\end{equation}
and output $Q_{t+1}/\Tr\left[Q_{t+1}\right]$. 
\end{itemize} 

Also, when applied to the optimization problem~\eqref{def:ClassicalAugMean}, the operator $T_{F}$ reduces to the following form:
\begin{equation}
    \label{eq:ProposedClassicalOperator}
    T_{f} \colon \mathbb{R}_{++}^d\to\mathbb{R}_{++}^d:u\mapsto\left(\sum_{j=1}^n w[j]\frac{a_j^{\alpha}}{\inner{a_j^{\alpha}}{u}}\right)^{(1-\alpha)/\alpha}.
\end{equation}
Our proposed iteration rule for solving the optimization problem~\eqref{def:ClassicalAugMean} is then given by
\begin{itemize}
    \item Let $q_1\in\mathrm{relint}\left(\Delta_d\right)$.
    \item For every $t \in \mathbb{N}$, compute 
    \begin{equation}
        \label{alg:ProposedClassicalIter}
        q_{t+1} = T_{f} \left(q_{t}^{1-\alpha} \right)^{1/(1-\alpha)},
    \end{equation}
    and output $q_{t+1}/\inner{\mathbf{1}_d}{q_{t+1}}$. 
\end{itemize} 

To verify the well-definedness of the operator $T_{F}$, we first observe that
for each $j$, since $A_j\in\mathcal{D}_d$, the denominator $\Tr\left[ A_j^\alpha U \right]$ is finite and positive 
for any $U\in\mathcal{B}_{d,++}$.
Combining with the assumption that $\sum_{j=1}^n A_j$ is full-rank, this implies that $T_F(U)$ is also full-rank. 
This ensures that the operator $T_{F}$ is well-defined. 
To understand why we restrict the domain of the operator $T_{F}$ to $\mathcal{B}_{d,++}$, suppose $U\in\mathcal{B}_{d,+}$ is not full rank.
In this case, the denominator $\Tr\left[ A_j^\alpha U \right]$ may become zero, rendering the definition of the operator $T_{F}$ ill-defined.

To evaluate the per-iteration time complexity of the proposed iteration rule \eqref{alg:ProposedIter}, we 
express it 
explicitly as follows (recall \eqref{eq:IntroAugMeanIter}):
\begin{align}
    \label{eq:DiscreteCP}
    Q_{t+1} = \left(\sum_{j=1}^n w[j]\frac{ A_j^\alpha}{ \Tr\left[ A_j^\alpha Q_t^{1-\alpha} \right] }\right)^{1/\alpha}.
\end{align}
The matrix powers $A_j^\alpha$ can be computed and stored before the first iteration begins.
Given matrices $Q_t^{1-\alpha}$ and $A_j^\alpha$ for all $j\in\Set{1,2,\dots,n}$, each $\Tr\left[ A_j^\alpha Q_t^{1-\alpha} \right]$ can be computed in $O ( d^2 )$ time. 
Computing $Q_t^{1-\alpha}$ and raising a matrix to the power $1 / \alpha$ each require $O ( d^3 )$ time. 
Consequently, the initialization time complexity is $O\left(n d^{3}\right)$, and the per-iteration time complexity is $O\left(d^{3} + n d^{2}\right)$.

\subsection{Convergence Analysis}
\label{subsec:Conv}

Below, we present our main theorem.
\begin{theorem}
    \label{thm:PetzConv}
    For any $\alpha\in(1/2,1)\cup(1,\infty)$, let $\Set{Q_t}_{t\in\mathbb{N}}$ be the sequence of iterates generated by the iteration rule \eqref{alg:ProposedIter}.
    Then, we have 
    \begin{align}
        F\left(\frac{Q_{T+1}}{\Tr\left[Q_{T+1}\right]}\right)-F\left(Q_{\star}\right)
        &\leq \abs{\frac{1}{\alpha-1}}d_{\mathrm{T}}\left(Q_{\star}^{1-\alpha},\left(\frac{Q_{T+1}}{\Tr\left[Q_{T+1}\right]}\right)^{1-\alpha}\right) \nonumber\\
        &\leq \left|\frac{2}{\alpha-1}\right|\cdot \left|1-\frac{1}{\alpha}\right|^T d_{\mathrm{T}}\left(Q_{\star}^{1-\alpha},Q_1^{1-\alpha}\right),
    \end{align}
    where $Q_{\star}$ is the minimizer of the optimization problem~\eqref{def:PetzAugMean}.
    Moreover, for $\alpha > 1$, the function values are non-increasing, i.e.,
    \begin{align}
        F\left(\frac{Q_{t+1}}{\Tr\left[Q_{t+1}\right]}\right)
        \leq
        F\left(\frac{Q_{t}}{\Tr\left[Q_{t}\right]}\right),\quad\forall t\in\mathbb{N}.
    \end{align}
    Furthermore, the quantity $d_{\mathrm{T}}\left(Q_{\star}^{1-\alpha},Q_1^{1-\alpha}\right)$ is bounded above.
\end{theorem}
\begin{remark}
    \label{remark:ClassicalConv}
    By viewing all vectors $a_j$ and $q_t$ as diagonal matrices, and following the proof of Theorem \ref{thm:PetzConv}, we obtain a similar convergence guarantee for the commuting case defined in \eqref{def:ClassicalAugMean}.
    In particular, for any $\alpha\in(1/2,1)\cup(1,\infty)$, let $\Set{q_t}_{t\in\mathbb{N}}$ be the sequence of iterates generated by the iteration rule \eqref{alg:ProposedClassicalIter}.
    Then, we have 
    \begin{align}
        f\left(\frac{q_{T+1}}{\inner{\mathbf{1}_d}{q_{T+1}}}\right)-f\left(q_{\star}\right)
        &\leq \abs{\frac{1}{\alpha-1}}d_{\mathrm{T}}\left(q_{\star}^{1-\alpha},\left(\frac{q_{T+1}}{\inner{q_{T+1}}{\mathbf{1}_d}}\right)^{1-\alpha}\right) \nonumber\\
        &= d_{\mathrm{T}}\left(q_{\star},\frac{q_{T+1}}{\inner{q_{T+1}}{\mathbf{1}_d}}\right) \nonumber\\
        &\leq \left|\frac{2}{\alpha-1}\right|\cdot \left|1-\frac{1}{\alpha}\right|^T d_{\mathrm{T}}\left(q_{\star}^{1-\alpha},q_1^{1-\alpha}\right)\nonumber\\
        &= 2 \left|1-\frac{1}{\alpha}\right|^T d_{\mathrm{T}}\left(q_{\star},q_1\right),
    \end{align}
    where $q_{\star}$ is the minimizer of the optimization problem~\eqref{def:ClassicalAugMean} and the two equalities follow from Lemma \ref{ineq:NegCurveProp}.
    Moreover, for $\alpha > 1$, the function values are non-increasing, i.e.,
    \begin{align}
        f\left(\frac{q_{t+1}}{\inner{\mathbf{1}_d}{q_{t+1}}}\right)
        \leq
        f\left(\frac{q_{t}}{\inner{\mathbf{1}_d}{q_{t}}}\right),\quad\forall t\in\mathbb{N}.
    \end{align}
    Furthermore, the quantity $d_{\mathrm{T}}\left(q_{\star},q_1\right)$ is bounded above.
\end{remark}

\paragraph{Roadmap.}
The proof of Theorem \ref{thm:PetzConv}, which we defer to Section \ref{subsec:PfMainThm}, relies on the following observations.
\begin{itemize}
    \item 
    The operator $T_{F}$ is contractive with a ratio of $| 1 - 1 / \alpha |$ with respect to the Thompson metric (see Section \ref{subsec:Contract}). 
    As a result, it has a unique fixed point, and the iterates $Q_t$ converge to this fixed point at a rate of $O\left(|1-1/\alpha|^T\right)$
    with respect to the Thompson metric. 
    \item The unique fixed point of $T_{F}$ 
    coincides with 
    the minimizer of the optimization problem~\eqref{def:PetzAugMean} (see Section \ref{subsec:FixPt}).
    \item The iterates $Q_t$ may not be ``physical,'' in the sense that they may not have unit traces.
	Fortunately, we show that the Thompson metrics between the iterates and the minimizer are preserved under trace normalization, up to a multiplicative constant. 
    Therefore, the trace-normalized iterates still converge 
    to the minimizer
    at a rate of $O\left(|1-1/\alpha|^T\right)$ (see Section \ref{subsec:NormalStable}).
    \item 
    The variation in function values can be upper-bounded by the Thompson metric between the iterates and the minimizer (see Section \ref{subsec:ControlFval}). 
    Consequently, the above error bound with respect to the Thompson metric translates into an error bound in function value.
\end{itemize}
In addition, we provide a counterexample in Section \ref{subsubsec:FailureAlphaSmall} to show that our iteration rule can fail to converge for $\alpha\in(0,1/2]$.

\subsubsection{Contractivity of $T_{F}$}
\label{subsec:Contract}
\begin{lemma}[{Contractive Property}]
    \label{ineq:Contract}
    Let $\alpha\in(1/2,1)\cup(1,\infty)$. 
    For any $U,V\in\mathcal{B}_{d,++}$, we have
    \begin{align}
        d_{\mathrm{T}}\left(T_{F}(V),T_{F}(U)\right)
        \leq \left|1-\frac{1}{\alpha}\right| d_{\mathrm{T}}\left(V,U\right).
    \end{align}
\end{lemma}
\begin{proof}
     By Definition \ref{eq:Thompson}, we have
    \begin{align}
        \exp\left(-d_{\mathrm{T}}\left(V,U\right)\right)V
        \leq U
        \leq \exp\left(d_{\mathrm{T}}\left(V,U\right)\right)V.\nonumber
    \end{align}
    Since $A_j\in\mathcal{D}_d$, we write
    \begin{align}
        T_{F}\left(U\right)^{\alpha/(1-\alpha)}
        &=\sum_{j=1}^n w[j]\frac{A_j^{\alpha}}{\Tr\left[A_j^{\alpha}U\right]}\nonumber\\
        &\geq\exp\left(- d_{\mathrm{T}}\left(V,U\right)\right)\sum_{j=1}^n w[j]\frac{A_j^{\alpha}}{\Tr\left[A_j^{\alpha}V\right]}\nonumber\\
        &=\exp\left(-d_{\mathrm{T}}\left(V,U\right)\right)T_{F}\left(V\right)^{\alpha/(1-\alpha)}.\nonumber
    \end{align}
    Similarly, we write
    \begin{align}
        T_{F}\left(U\right)^{\alpha/(1-\alpha)}
        &\leq\exp\left(d_{\mathrm{T}}\left(V,U\right)\right)\sum_{j=1}^n w[j]\frac{A_j^{\alpha}}{\Tr\left[A_j^{\alpha}V\right]}\nonumber\\
        &=\exp\left(d_{\mathrm{T}}\left(V,U\right)\right)T_{F}\left(V\right)^{\alpha/(1-\alpha)}.\nonumber
    \end{align}
    By Definition \ref{eq:Thompson}, the two inequalities above imply that 
    \begin{align}
        d_{\mathrm{T}}\left(T_{F}\left(V\right)^{\alpha/(1-\alpha)},T_{F}\left(U\right)^{\alpha/(1-\alpha)}\right) \leq d_{\mathrm{T}} ( V, U ).
    \end{align}
    Then, by Lemma \ref{ineq:NegCurveProp}, we have
    \begin{align}
        d_{\mathrm{T}}\left(T_{F}\left(V\right),T_{F}\left(U\right)\right)
        &=d_{\mathrm{T}}\left(\left(T_{F}\left(V\right)^{\alpha/(1-\alpha)}\right)^{(1-\alpha)/\alpha},\left(T_{F}\left(U\right)^{\alpha/(1-\alpha)}\right)^{(1-\alpha)/\alpha}\right)\nonumber\\
        &\leq\left|\frac{1-\alpha}{\alpha}\right|d_{\mathrm{T}}\left(T_{F}\left(V\right)^{\alpha/(1-\alpha)},T_{F}\left(U\right)^{\alpha/(1-\alpha)}\right)\nonumber\\
        &\leq\left|\frac{1-\alpha}{\alpha}\right|d_{\mathrm{T}}\left(V,U\right).\nonumber
    \end{align}
\end{proof}
\subsubsection{Fixed-Point Property of $T_{F}$}
\label{subsec:FixPt}
\begin{lemma}
    \label{eq:FixPtIsMin}
    For any $\alpha  \in (0, 1) \cup (1, \infty)$, there exists a unique minimizer $Q_{\star}$ of the optimization problem~\eqref{def:PetzAugMean}.
    Moreover, for the same $Q_{\star}$, $Q_{\star}^{1-\alpha}$ is the unique fixed point of the operator $T_{F}$ for $\alpha\in(0,1)\cup(1,\infty)$.
\end{lemma}
\paragraph{Roadmap.}
The remainder of this section is devoted to the proof of Lemma \ref{eq:FixPtIsMin}.
For $\alpha\in(0,1)$, Lemma \ref{eq:FixPtIsMin} has been 
proven 
by \citet[Proposition 2(b)]{Cheng2019}.
For $\alpha\in(1,\infty)$, the proof of Lemma \ref{eq:FixPtIsMin} relies on the following observations:
\begin{itemize}
    \item The traces of the iterates are always less than or equal to $1$ (Lemma \ref{ineq:TraceBound}).
    \item The function values are non-increasing (Lemma \ref{ineq:FvalDecrease}).
\end{itemize}

We will use Lemma \ref{ineq:Araki} and Lemma \ref{ineq:HolderEq} to prove Lemma \ref{ineq:TraceBound}.
\begin{lemma}[{Araki-Lieb-Thirring Inequality \citep{Araki1990}}]
\label{ineq:Araki}
For any $U,V\in\mathcal{B}_{d,++}$, we have 
\begin{align}
    \Tr\left[\left(V^{1/2}UV^{1/2}\right)^{sr}\right]\leq\Tr\left[\left(V^{r/2} U^r V^{r/2}\right)^{s}\right],
\end{align}
for all $s>0$ and $r\geq 1$.
\end{lemma}
\begin{lemma}[{H\"older Inequality \citep{Larotonda2018}}]
    \label{ineq:HolderEq}
    For any $U,V\in\mathcal{B}_{d,+}\backslash\Set{0}$ and $p>1$, we have
    \begin{align}
        \Tr[UV]\leq\Tr[U^{p}]^{1/p}\Tr[V^{p/(p-1)}]^{1-1/p}.
    \end{align}
    Moreover, equality holds if and only if 
    \begin{align}
        \frac{U^p}{\Tr[U^p]}=\frac{V^{p/(p-1)}}{\Tr[V^{p/(p-1)}]}.
    \end{align}
\end{lemma}
\begin{lemma}[{Bound of Trace}]
\label{ineq:TraceBound}
For any $\alpha\in(1,\infty)$ and $Q\in\mathcal{B}_{d,++}$ such that $\Tr\left[Q\right]\leq 1$, we have 
\begin{align}
    \Tr\left[T_{F}\left(Q^{1-\alpha}\right)^{1/(1-\alpha)}\right]\leq 1.
\end{align}
Moreover, equality holds if and only if 
$Q$ is a fixed point of $T_{F}\left((\cdot)^{1-\alpha}\right)^{1/(1-\alpha)}$ on $\mathcal{D}_d$.
\end{lemma}

\begin{proof}
Let $Q \in \mathcal{B}_{d,++}$ such that $\Tr\left[Q\right]\leq 1$. 
Let $U = T_{F}\left(Q^{1-\alpha}\right)^{\alpha/(1-\alpha)}$ and $V = Q^{1-\alpha}$.
Then, both $U$ and $V$ are positive definite, and we have 
\begin{align}
    \Tr \left[ UV \right]
    = \Tr\left[ \left(\sum_{j=1}^n w[j]\frac{ A_j^\alpha}{ \Tr\left[ A_j^\alpha Q^{1-\alpha} \right] }\right) Q^{1-\alpha} \right]
    =1.\nonumber
\end{align}
Then, we write 
    \begin{align}
        \Tr\left[\left(T_{F}\left(Q^{1-\alpha}\right)\right)^{1/(1-\alpha)}\right]
        &= \Tr\left[ U^{ 1/\alpha }\right]\\
        &= \Tr\left[ \left( V^{ 1/(2\alpha)} U^{ 1/\alpha } V^{ 1/(2\alpha) }\right) V^{ - 1/\alpha } \right]\\
        &\leq\Tr\left[ \left( V^{ 1/(2\alpha) } U^{ 1/\alpha } V^{ 1/(2\alpha) }\right)^\alpha \right]^{ 1/\alpha }  \Tr\left[ \left( V^{ -1/\alpha } \right)^{\alpha/(\alpha-1)} \right]^{ 1- 1/\alpha } \\
        &=\Tr\left[ \left( V^{  1/(2\alpha) } U^{1/\alpha } V^{ 1/(2\alpha) }\right)^\alpha \right]^{ 1/\alpha }  \Tr\left[ Q \right]^{ 1- 1/\alpha }\\
        &\leq \Tr\left[ \left( V^{ 1/(2\alpha)} U^{ 1/\alpha } V^{ 1/(2\alpha) }\right)^\alpha \right]^{ 1/\alpha },
    \end{align}
    where the first inequality follows from the H\"older inequality (Lemma \ref{ineq:HolderEq}), 
    the third equality follows from the definition of $V$,
    and the last inequality follows from the assumption that $\Tr[Q] \leq 1$.

    Then, by the Araki-Lieb-Thirring inequality (Lemma \ref{ineq:Araki}), we have 
    \begin{align}
        \Tr\left[ \left( V^{ 1/(2\alpha) } U^{ 1/\alpha } V^{ 1/(2\alpha) }\right)^\alpha \right] 
        \leq \Tr\left[ V^{1/2 } U V^{ 1/2 }\right]
        = \Tr\left[UV\right]= 1.
    \end{align} 
    Therefore, we conclude that 
    \begin{align}
        \Tr\left[\left(T_{F}\left(Q^{1-\alpha}\right)\right)^{1/(1-\alpha)}\right]\leq 1.
    \end{align}
    
    We proceed to prove the if and only if condition. 
    Note that the ``if'' direction holds trivially. 
    It remains to prove the ``only if'' direction. 
    Suppose that $\Tr\left[T_{F}\left(Q^{1-\alpha}\right)^{1/(1-\alpha)}\right]=1$.
    Let $U$ and $V$ be defined as above. 
    Recall that we have proved
    \begin{align}
        \Tr\left[T_{F}\left(Q^{1-\alpha}\right)^{1/(1-\alpha)}\right]
        &= \Tr\left[ \left( V^{ 1/(2\alpha)} U^{ 1/\alpha } V^{ 1/(2\alpha) }\right) V^{ - 1/\alpha } \right]\\
        &\leq\Tr\left[ \left( V^{ 1/(2\alpha) } U^{ 1/\alpha } V^{ 1/(2\alpha) }\right)^\alpha \right]^{ 1/\alpha }  \Tr\left[ \left( V^{ -1/\alpha } \right)^{\alpha/(\alpha-1)} \right]^{ 1- 1/\alpha } \\
        &=\Tr\left[ \left( V^{ 1/(2\alpha) } U^{ 1/\alpha } V^{ 1/(2\alpha) }\right)^\alpha \right]^{ 1/\alpha }  \Tr\left[ Q \right]^{ 1- 1/\alpha },
    \end{align}
    and
    \begin{align}
        \Tr\left[ \left( V^{ 1/(2\alpha) } U^{ 1/\alpha } V^{ 1/(2\alpha) }\right)^\alpha \right]\leq 1.
    \end{align}
    Since we have assumed that $\Tr\left[T_{F}\left(Q^{1-\alpha}\right)^{1/(1-\alpha)}\right]=1$ and $\Tr\left[ Q \right]\leq 1$,
    it must be the case that
    \begin{align}
        \Tr\left[ \left( V^{ 1/(2\alpha) } U^{ 1/\alpha } V^{ 1/(2\alpha) }\right)^\alpha \right]
        =\Tr\left[ \left( V^{ -1/\alpha } \right)^{\alpha/(\alpha-1)} \right]
        =\Tr[Q]
        =1,
    \end{align}
    and
    \begin{align}
        &\Tr\left[ \left( V^{ 1/(2\alpha)} U^{ 1/\alpha } V^{ 1/(2\alpha) }\right) V^{ - 1/\alpha } \right]\\
        &\quad=\Tr\left[ \left( V^{ 1/(2\alpha) } U^{ 1/\alpha } V^{ 1/(2\alpha) }\right)^\alpha \right]^{ 1/\alpha }  \Tr\left[ \left( V^{ -1/\alpha } \right)^{\alpha/(\alpha-1)} \right]^{ 1- 1/\alpha }.
    \end{align}    
    Using the equality condition of the H\"{o}lder inequality (Lemma \ref{ineq:HolderEq}), the above equality implies
    \begin{align}
        \frac{\left( V^{ 1/(2\alpha) } U^{ 1/\alpha } V^{ 1/(2\alpha) }\right)^\alpha}{\Tr\left[\left( V^{ 1/(2\alpha) } U^{ 1/\alpha } V^{ 1/(2\alpha) }\right)^\alpha\right]}
        =\frac{\left( V^{ -1/\alpha } \right)^{\alpha/(\alpha-1)}}{\Tr\left[ \left( V^{ -1/\alpha } \right)^{\alpha/(\alpha-1)} \right]},
    \end{align}
    where the denominators on both sides, as concluded above, are equal to $1$.
    Therefore, we have  
    \[
    \left( V^{ 1/(2\alpha) } U^{ 1/\alpha } V^{ 1/(2\alpha) }\right)^\alpha=\left( V^{ -1/\alpha } \right)^{\alpha/(\alpha-1)} . 
    \]
    Plugging in the definitions of $U$ and $V$, we get 
    \[
    T_{F}\left(Q^{1-\alpha}\right)^{1/(1-\alpha)}=Q . 
    \]
    This completes the proof. 
\end{proof}

Next, using Lemma \ref{ineq:TraceBound}, we prove that the function values are non-increasing, as stated in Lemma \ref{ineq:FvalDecrease}.
\begin{lemma}[{Monotonicity of the Function Value}]
    \label{ineq:FvalDecrease}
    For any $\alpha\in(1,\infty)$ and $Q\in\mathcal{B}_{d,++}$ such that $\Tr\left[Q\right]\leq 1$, we have  
    \begin{align}
        F\left(T_{F}\left(Q^{1-\alpha}\right)^{1/(1-\alpha)}\right)
        \leq F(Q).
    \end{align}
    Moreover, equality holds if and only if $Q^{1 - \alpha}$ is a fixed point of $T_{F}$.
\end{lemma}
\begin{proof}
    Let $Q \in \mathcal{B}_{d,++}$ such that $\Tr\left[Q\right]\leq 1$. 
    We write
    \begin{align}
        &F\left(T_{F}\left(Q^{1-\alpha}\right)^{1/(1-\alpha)}\right)- F(Q)\nonumber\\
        &\quad =\frac{1}{\alpha-1}\sum_{j=1}^n w[j]\log\left(\frac{\Tr\left[A_j^{\alpha}T_{F}\left(Q^{1-\alpha}\right)\right]}{\Tr\left[A_j^{\alpha}Q^{1-\alpha}\right]}\right)\nonumber\\
        &\quad\leq\frac{1}{\alpha-1}\sum_{j=1}^n w[j]\left(\frac{\Tr\left[A_j^{\alpha}T_{F}\left(Q^{1-\alpha}\right)\right]}{\Tr\left[A_j^{\alpha}Q^{1-\alpha}\right]}-1\right)\nonumber\\
        &\quad =\frac{1}{\alpha-1}\left(\Tr\left[T_{F}\left(Q^{1-\alpha}\right)^{\alpha/(1-\alpha)}T_{F}\left(Q^{1-\alpha}\right)\right]-1\right),\nonumber 
    \end{align}
    where the first inequality exploits the fact that $\log x \leq x - 1$, 
    and 
    the second equality follows from the definition of the operator $T_{F}$.
    The lemma then follows from Lemma \ref{ineq:TraceBound}. 
\end{proof}

Finally, we prove Lemma \ref{eq:FixPtIsMin} by showing that the unique fixed point of the operator $T_{F}$ is also the minimizer of the optimization problem~\eqref{def:PetzAugMean}.
\begin{proof}(Lemma \ref{eq:FixPtIsMin})
    By Lemma \ref{eq:ThompIsMetric},  Lemma \ref{ineq:Contract}, and the Banach fixed point theorem, there exists a unique $\tilde{Q}_{\star}\in\mathcal{B}_{d,++}$ such that $\tilde{Q}_{\star}^{1-\alpha}$ is the fixed point of the operator $T_{F}$ for $\alpha\in(1/2,1)\cup(1,\infty)$.
    We recall that for $\alpha\in(0,1)$, Lemma \ref{eq:FixPtIsMin} has already been proved by \citet[Proposition 2(b)]{Cheng2019}.
    For $\alpha\in(1,\infty)$, let $Q_{\star}$ be the minimizer of the optimization problem~\eqref{def:PetzAugMean}.
    Suppose that $Q_{\star}^{1-\alpha}$ is not the fixed point of $T_{F}$.
    Then, the equality conditions in Lemma \ref{ineq:TraceBound} and Lemma \ref{ineq:FvalDecrease} do not hold, and we have 
    \begin{align}
        &F\left(\frac{T_{F}\left(Q_{\star}^{1-\alpha}\right)^{1/(1-\alpha)}}{\Tr\left[T_{F}\left(Q_{\star}^{1-\alpha}\right)^{1/(1-\alpha)}\right]}\right)\\
        &\quad=F\left(T_{F}\left(Q_{\star}^{1-\alpha}\right)^{1/(1-\alpha)}\right)+\log\left(\Tr\left[T_{F}\left(Q_{\star}^{1-\alpha}\right)^{1/(1-\alpha)}\right]\right)\\
        &\quad<F\left(Q_{\star}\right).
    \end{align}
    This inequality contradicts the optimality of $Q_{\star}$.
    Therefore, we conclude that $Q_{\star}=\tilde{Q}_{\star}$.
\end{proof}
\begin{remark}
    We note that the fixed-point property proven by \citet[Proposition 4(c)]{Cheng2018}\footnote{Note that this result does not appear in the journal version \citep{Cheng2022}.} can also lead to the conclusion in Lemma \ref{eq:FixPtIsMin}.
    However, our proof strategy differs from theirs.
    Furthermore, for $\alpha > 1$, our proof strategy yields an additional useful property for implementing the proposed iteration rule: the function values are non-increasing along the iteration path (Lemma \ref{ineq:FvalDecrease}).
\end{remark}


\subsubsection{Preservation of Thompson Metric under Trace-Normalization}
\label{subsec:NormalStable}
The constraint set of the optimization problem~\eqref{def:PetzAugMean} is the set of quantum states $\mathcal{D}_d$, whereas the traces of the iterates $Q_t$ may not equal $1$.
To address this, we show in Lemma \ref{ineq:SmallNormalThompson} that the Thompson metric is 
preserved
under 
trace-normalization, up to a multiplicative constant of $2$. 

\begin{lemma}
    \label{ineq:SmallNormalThompson}
    Let $\alpha\in(0,1)\cup(1,\infty)$.
    For any $U,V\in\mathcal{B}_{d,++}$ such that $\Tr[V]=1$, we have
    \begin{align}
        d_{\mathrm{T}}\left(V^{1-\alpha},\left(\frac{U}{\Tr[U]}\right)^{1-\alpha}\right)
        \leq 2 d_{\mathrm{T}}\left(V^{1-\alpha},U^{1-\alpha}\right).
    \end{align}
\end{lemma}

We will use 
Lemma \ref{ineq:Lidskii} to prove Lemma \ref{ineq:SmallNormalThompson}. 

\begin{lemma}[{\citep[Corollary 7.7.4(c)]{Horn2013}}]
    \label{ineq:Lidskii}
    For any $U,V\in\mathcal{B}_{d,++}$ such that $U\leq V$, we have
    \begin{align}
        \lambda_i(U)\leq\lambda_i(V),\quad\forall i\in\set{1,2,\dots,d}.
    \end{align}
\end{lemma}
\begin{proof}(Lemma \ref{ineq:SmallNormalThompson})
    Let $U,V$ be as defined in Lemma \ref{ineq:SmallNormalThompson}.
    By Lemma \ref{ineq:ThompsonLogHomo}, we write 
    \begin{align}
        d_{\mathrm{T}}\left(V^{1-\alpha},\left(\frac{U}{\Tr[U]}\right)^{1-\alpha}\right)\leq d_{\mathrm{T}}\left(V^{1-\alpha},U^{1-\alpha}\right)+|\alpha-1|\cdot\left|\log\left(\Tr[U]\right)\right|.
    \end{align}
    It remains to bound the quantity $\Tr[U]$, which can be written as follows:
    \begin{align}
        \Tr[U]
        =\sum_{i=1}^{d}\left(\lambda_i(U)^{1-\alpha}\right)^{1/(1-\alpha)}
        =\sum_{i=1}^{d}\left(\lambda_i(U^{1-\alpha})\right)^{1/(1-\alpha)}.
    \end{align}
    On the other hand, by Lemma \ref{ineq:Lidskii}, for each $i\in\set{1,2,\dots,d}$, we have
    \begin{align}
        \exp\left(- d_{\mathrm{T}}
        \left(V^{1-\alpha},U^{1-\alpha}\right)\right)\lambda_i\left(V^{1-\alpha}\right)
        \leq \lambda_i\left(U^{1-\alpha}\right)
        \leq\exp\left( d_{\mathrm{T}}
        \left(V^{1-\alpha},U^{1-\alpha}\right)\right)\lambda_i\left(V^{1-\alpha}\right).
    \end{align}
    Consequently, we obtain
    \begin{align}
        \exp\left( \frac{-d_{\mathrm{T}}
        \left(V^{1-\alpha},U^{1-\alpha}\right)}{|1-\alpha|}\right)\Tr[V]
        \leq\Tr[U]
        \leq \exp\left( \frac{d_{\mathrm{T}}
        \left(V^{1-\alpha},U^{1-\alpha}\right)}{|1-\alpha|}\right)\Tr[V].
    \end{align}
    Since we assume that $\Tr[V]=1$, it follows that
    \begin{align}
        \abs{\log\left(\Tr[U]\right)}
        \leq\left|\frac{1}{\alpha-1}\right|d_{\mathrm{T}}
        \left(V^{1-\alpha},U^{1-\alpha}\right).
    \end{align}
    This concludes the proof.
\end{proof}

\subsubsection{Bounding Variation in Function Values}
\label{subsec:ControlFval}
Finally, it remains to translate the convergence guarantee of the iterates into that of the function values.
We prove that the difference between the function values is bounded above by the Thompson metric.
\begin{lemma}
    \label{ineq:SmallThomp2SmallOptError}
    Let $\alpha\in(0,1)\cup(1,\infty)$.
    For any $U,V\in\mathcal{B}_{d,++}$, we have
    \begin{align}
        F\left(U\right)-F\left(V\right)\leq \left|\frac{1}{\alpha-1}\right|d_{\mathrm{T}}\left(V^{1-\alpha},U^{1-\alpha}\right).
    \end{align}
\end{lemma}
\begin{proof}(Lemma \ref{ineq:SmallThomp2SmallOptError})
    Let $U,V$ be as defined in Lemma \ref{ineq:SmallThomp2SmallOptError}.
    By Definition \ref{eq:Thompson}, we have
    \begin{align}
        \exp\left(- d_{\mathrm{T}}
        \left(V^{1-\alpha},U^{1-\alpha}\right)\right)V^{1-\alpha}
        \leq U^{1-\alpha}
        \leq\exp\left( d_{\mathrm{T}}
        \left(V^{1-\alpha},U^{1-\alpha}\right)\right)V^{1-\alpha}.
    \end{align}
    Therefore, we write
    \begin{align}
        F\left(U\right)
        &=\frac{1}{\alpha-1}\sum_{j=1}^n w[j]\log\operatorname{Tr}\left[ A_j^{\alpha} U^{1-\alpha}\right]\nonumber\\
        &\leq
        \frac{1}{\alpha-1}\sum_{j=1}^n w[j]\log\operatorname{Tr}\left[ A_j^{\alpha}\left(\exp\left( d_{\mathrm{T}}\left(V^{1-\alpha},U^{1-\alpha}\right)\right)V^{1-\alpha}\right)\right]\nonumber\\
        &=F\left(V\right)+\frac{1}{\alpha-1}d_{\mathrm{T}}\left(V^{1-\alpha},U^{1-\alpha}\right)\nonumber
    \end{align}
    for $\alpha>1$, and
    \begin{align}
        F\left(U\right)
        &\leq
        \frac{1}{\alpha-1}\sum_{j=1}^n w[j]\log\operatorname{Tr}\left[ A_j^{\alpha}\left(\exp\left(- d_{\mathrm{T}}\left(V^{1-\alpha},U^{1-\alpha}\right)\right)V^{1-\alpha}\right)\right]\nonumber\\
        &=F\left(V\right)+\frac{1}{1-\alpha} d_{\mathrm{T}}\left(V^{1-\alpha},U^{1-\alpha}\right)\nonumber
    \end{align}
    for $\alpha<1$.
    This concludes the proof.
    
\end{proof}

\subsubsection{Proof of the Main Theorem}
\label{subsec:PfMainThm}
\begin{proof}(Theorem \ref{thm:PetzConv})
    Let $Q_{\star}$ and $Q_t$ be defined as in Theorem \ref{thm:PetzConv}.
    By Lemma \ref{ineq:Contract}, Lemma \ref{eq:FixPtIsMin}, and induction, we write
    \begin{align}
        d_{\mathrm{T}}\left(Q_{\star}^{1-\alpha},Q_{T+1}^{1-\alpha}\right) 
        \leq\left|1-\frac{1}{\alpha}\right|^T d_{\mathrm{T}}\left(Q_{\star}^{1-\alpha},Q_1^{1-\alpha}\right).
    \end{align}
    Consequently, by Lemma \ref{ineq:SmallNormalThompson}, we obtain
    \begin{align}
        d_{\mathrm{T}}\left(Q_{\star}^{1-\alpha},\left(\frac{Q_{T+1}}{\Tr\left[Q_{T+1}\right]}\right)^{1-\alpha}\right) 
        \leq 2\left|1-\frac{1}{\alpha}\right|^T d_{\mathrm{T}}\left(Q_{\star}^{1-\alpha},Q_1^{1-\alpha}\right).
    \end{align}
    By Lemma \ref{ineq:SmallThomp2SmallOptError}, it follows that
    \begin{align}
        F\left(\frac{Q_{T+1}}{\Tr\left[Q_{T+1}\right]}\right)-F\left(Q_{\star}\right)
        &\leq \abs{\frac{1}{\alpha-1}}d_{\mathrm{T}}\left(Q_{\star}^{1-\alpha},\left(\frac{Q_{T+1}}{\Tr\left[Q_{T+1}\right]}\right)^{1-\alpha}\right) \nonumber\\
        &\leq \left|\frac{2}{\alpha-1}\right|\cdot \left|1-\frac{1}{\alpha}\right|^T d_{\mathrm{T}}\left(Q_{\star}^{1-\alpha},Q_1^{1-\alpha}\right).
    \end{align}
    It remains to prove that
    \begin{align}
        F\left(\frac{Q_{t+1}}{\Tr\left[Q_{t+1}\right]}
        \right)
        \leq F\left(\frac{Q_{t}}{\Tr\left[Q_{t}\right]}\right),\quad\forall\alpha>1.
    \end{align} 
    By the monotonicity of the function value (Lemma \ref{ineq:FvalDecrease}) and the bound of the trace (Lemma \ref{ineq:TraceBound}), we write
    \begin{align}
        &F\left(\frac{T_{F}\left(\left(\frac{Q_t}{\Tr\left[Q_t\right]}\right)^{1-\alpha}\right)^{1/(1-\alpha)}}{\Tr\left[T_{F}\left(\left(\frac{Q_t}{\Tr\left[Q_t\right]}\right)^{1-\alpha}\right)^{1/(1-\alpha)}\right]}\right)\\
        &\quad=F\left(T_{F}\left(\left(\frac{Q_t}{\Tr\left[Q_t\right]}\right)^{1-\alpha}\right)^{1/(1-\alpha)}\right)+\log\left(\Tr\left[T_{F}\left(\left(\frac{Q_t}{\Tr\left[Q_t\right]}\right)^{1-\alpha}\right)^{1/(1-\alpha)}\right]\right)\\
        &\quad\leq F\left(\frac{Q_t}{\Tr\left[Q_t\right]}\right)+0,\quad\forall\alpha>1.
    \end{align}
    It remains to prove that
    \begin{align}
        \frac{Q_{t+1}}{\Tr\left[Q_{t+1}\right]}
        = \frac{T_{F}\left(\left(\frac{Q_t}{\Tr\left[Q_t\right]}\right)^{1-\alpha}\right)^{1/(1-\alpha)}}{\Tr\left[T_{F}\left(\left(\frac{Q_t}{\Tr\left[Q_t\right]}\right)^{1-\alpha}\right)^{1/(1-\alpha)}\right]}.
    \end{align}
    Note that for any $Q\in\mathcal{B}_{d,++}$ and $\gamma>0$, we have
    \begin{align}
        T_{F}\left(\left(\gamma Q\right)^{1-\alpha}\right)^{1/(1-\alpha)}
        =\gamma^{(\alpha-1)/\alpha} T_{F}\left(Q^{1-\alpha}\right)^{1/(1-\alpha)},
    \end{align}
    and thus,
    \begin{align}
        \frac{T_{F}\left(\left(\gamma Q\right)^{1-\alpha}\right)^{1/(1-\alpha)}}{\Tr\left[T_{F}\left(\left(\gamma Q\right)^{1-\alpha}\right)^{1/(1-\alpha)}\right]}
        =\frac{T_{F}\left(Q^{1-\alpha}\right)^{1/(1-\alpha)}}{\Tr\left[T_{F}\left(Q^{1-\alpha}\right)^{1/(1-\alpha)}\right]}.
    \end{align}
    This implies that, by taking $Q=Q_t$ and $\gamma=1/\Tr\left[Q_t\right]$,
    \begin{align}
        \frac{Q_{t+1}}{\Tr\left[Q_{t+1}\right]}
        =\frac{T_{F}\left(\left(Q_t\right)^{1-\alpha}\right)^{1/(1-\alpha)}}{\Tr\left[T_{F}\left(\left(Q_t\right)^{1-\alpha}\right)^{1/(1-\alpha)}\right]}
        = \frac{T_{F}\left(\left(\frac{Q_t}{\Tr\left[Q_t\right]}\right)^{1-\alpha}\right)^{1/(1-\alpha)}}{\Tr\left[T_{F}\left(\left(\frac{Q_t}{\Tr\left[Q_t\right]}\right)^{1-\alpha}\right)^{1/(1-\alpha)}\right]}.
    \end{align}
    
    Finally, since $Q_1^{1-\alpha}$ and $Q_{\star}^{1-\alpha}$ are full-rank matrices, and the Thompson metric is a metric on $\mathcal{B}_{d,++}$ (Lemma \ref{eq:ThompIsMetric}), the quantity $d_{\mathrm{T}}\left(Q_{\star}^{1-\alpha},Q_1^{1-\alpha}\right)$ is finite. 
    This concludes the proof.
\end{proof}
\subsubsection{Failure for $\alpha\in(0,1/2]$}
\label{subsubsec:FailureAlphaSmall}
It is natural to ask whether our iteration rule \eqref{alg:ProposedIter} converges for $\alpha\in(0,1/2]$, which is not covered by Theorem \ref{thm:PetzConv}.
We provide a counterexample: 
let $n=2$, $w[1]=w[2]=1/2$, and
\begin{equation}
    A_1
    =\begin{bmatrix}
    1 & 0\\
    0 & 0
    \end{bmatrix},
    \quad A_2
    =\begin{bmatrix}
    0 & 0\\
    0 & 1
    \end{bmatrix}.
\end{equation}
In this case, by the fixed-point property established in Lemma \ref{eq:FixPtIsMin}, it can be verified that 
\begin{align}
    Q_{\star}
    =\begin{bmatrix}
    0.5 & 0\\
    0 & 0.5
    \end{bmatrix}.
\end{align}

Consider the following initialization:
\begin{align}
    Q_{1}
    =\begin{bmatrix}
    q_1[1] & 0\\
    0 & q_1[2]
    \end{bmatrix},
\end{align}
where $q_1[1],q_1[2]>0$, $q_1[1]+q_1[2]=1$, and $q_1[1]\neq q_1[2]$.
The iterates $Q_t$ generated by the iteration rule \eqref{alg:ProposedIter} remain diagonal; their diagonal elements can be represented as $q_t=(q_t[1],q_t[2])$, where
\begin{align}
    q_{t+1}[i]
    =\left(\frac{1}{2}q_t[i]^{\alpha-1}\right)^{1/\alpha},\quad\forall t.
\end{align}
Consequently, we have
\begin{align}
    \abs{\log\left(\frac{q_{t+1}[1]}{q_{t+1}[2]}\right)}
    =\abs{\frac{1}{\alpha}-1}\cdot\abs{\log\left(\frac{q_t[1]}{q_t[2]}\right)},
\end{align}
showing that the ratio $q_t[1]/q_t[2]$ does not converge to $1$ for $\alpha\in(0,1/2]$, whereas the corresponding ratio for $Q_{\star}$ is $1$.
This provides a counterexample showing that the iteration rule \eqref{alg:ProposedIter} may fail to converge for $\alpha\in(0,1/2]$.

\section{Application to Petz-Augustin Capacity}
\label{sec:Capacity}
This section is devoted to applying our proposed algorithm \eqref{alg:ProposedIter} to develop the first iterative method with a non-asymptotic convergence guarantee for computing the Petz-Augustin capacity.
Section \ref{subsec:CapacityDef} introduces the problem formulation for computing the Petz-Augustin capacity, along with its application in classical-quantum channel coding.
In Section \ref{subsec:CapacityIter}, we apply the computation of the Petz-Augustin mean to develop an iterative method for computing the Petz-Augustin capacity.
Finally, in Section \ref{subsec:CapacityConv}, we analyze the convergence rate of the iterative method.

\subsection{Introduction}
\label{subsec:CapacityDef}
We consider the problem of computing the Petz-Augustin capacity $C_{\alpha}$ of order $\alpha\in(0,1)$, given by \citep{Cheng2019}
\begin{equation}
    \label{def:Capacity}
    C_{\alpha}
    \coloneqq -\min_{w\in\Delta_{n}}g(w),\quad g(w)
    \coloneqq -\sum_{j=1}^n w[j] D_{\alpha}\left(A_j\Vert Q_{\star}(w)\right), \quad Q_{\star}(w)
    \in \argmin_{Q\in\mathcal{D}_d}\sum_{j=1}^n w[j] D_{\alpha}\left(A_j\Vert Q\right),
\end{equation}
where the matrices $A_j$ are defined as in \eqref{def:PetzAugMean}.
The Petz-Augustin capacity generalizes the quantum channel capacity and characterizes the optimal error exponent in classical-quantum channel coding, quantifying the exponential rate at which the average error probability vanishes as the length of the message transmitted through the channel increases.
Specifically, although a closed-form expression is generally unavailable, the optimal error exponent $E(R)$ can be bounded as follows \citep{Dalai2013,Dalai2014,Cheng2019,Renes2025}:
\begin{equation}
    \sup_{\alpha\in(1/2,1)}\frac{1-\alpha}{\alpha}\left(C_{\alpha}-R\right)
    \leq E(R)
    \leq \sup_{\alpha\in(0,1)}\frac{1-\alpha}{\alpha}\left(C_{\alpha}-R\right),
\end{equation}
where $R>0$ is the coding rate in classical-quantum channel coding.
The lower and upper bounds coincide when $R\geq C_{1/2}$, making the characterization tight.
These bounds give operational meaning to the computation of the Petz-Augustin capacity $C_{\alpha}$ of order $\alpha\in(0,1)$.

However, we are not aware of any existing algorithm that computes $C_{\alpha}$ for any $\alpha\in(0,1)$ with a non-asymptotic convergence guarantee.
To this end, in Section \ref{subsec:CapacityIter}, we propose an iterative method \eqref{def:ProposedIterCapacity} for computing $C_{\alpha}$ for $\alpha\in(1/2,1)$, which uses our proposed iteration rule \eqref{alg:ProposedIter} for computing the Petz-Augustin mean as a subroutine.
In Section \ref{subsec:CapacityConv}, we prove that our method \eqref{def:ProposedIterCapacity} converges at a rate of $O\left(\log(n)/T\right)$, representing the first algorithm for computing the Petz-Augustin capacity with a non-asymptotic convergence guarantee.

\subsection{Iteration Rule}
\label{subsec:CapacityIter}
\paragraph{Roadmap.} This section is devoted to proposing a gradient-based method for computing the Petz-Augustin capacity $C_{\alpha}$ of order $\alpha\in(1/2,1)$.
Since the zeroth- and first-order oracles, $g(\cdot)$ and $\nabla g(\cdot)$, do not admit closed-form expressions,
we first show that, using Lemma \ref{lemma:Danskin} and Lemma \ref{lemma:ApproxGrad}, together with Theorem \ref{thm:PetzConv}, both $g(w)$ and $\nabla g(w)$ can nevertheless be efficiently computed via the iteration rule \eqref{alg:ProposedIter} (see Remark \ref{remark:ApproxGrad}).
Then, we use Lemma \ref{lemma:HessianBd} and Lemma \ref{lemma:MirrorConv} to motivate our choice of the gradient-based method \eqref{def:ProposedIterCapacity}.

\begin{lemma}
    \label{lemma:Danskin}
    For any $w\in\mathbb{R}_{++}^n$ and $\alpha\in(0,1)$, we have 
    \begin{equation}
        \nabla g(w)[j]
        =-D_{\alpha}\left(A_j\Vert Q_{\star}(w)\right),\quad\forall j.
    \end{equation}
\end{lemma}
The proof of Lemma \ref{lemma:Danskin}, which is deferred to Appendix \ref{sec:DanskinProof}, extends the work of \citet[Equation (25)]{Cheng2024b} to the non-commutative case.
\begin{lemma}
    \label{lemma:ApproxGrad}
    For any $w\in\mathrm{relint}\left(\Delta_n\right)$, $Q\in\mathcal{B}_{d,++}$, and $\alpha\in(0,1)$, let $\hat{g}_w=-\sum_{j=1}^n w[j] D_{\alpha}\left(A_j\Vert Q\right)$ and define $\hat{\nabla}\in\mathbb{R}^n$ by $\hat{\nabla}[j]\coloneqq-D_{\alpha}\left(A_j\Vert Q\right)$ for all $j$.
    Then, we have
    \begin{equation}
        \abs{\hat{g}_w-g(w)}
        \leq \abs{\frac{1}{1-\alpha}}d_{\mathrm{T}}\left(Q_{\star}(w)^{1-\alpha},Q^{1-\alpha}\right),
    \end{equation}
    and 
    \begin{equation}
        \abs{\hat{\nabla}[j]-\nabla g(w)[j]}
        \leq \abs{\frac{1}{1-\alpha}}d_{\mathrm{T}}\left(Q_{\star}(w)^{1-\alpha},Q^{1-\alpha}\right),\quad\forall j.
    \end{equation}
\end{lemma}
The proof of Lemma \ref{lemma:ApproxGrad}, which relies on Lemma \ref{lemma:Danskin} and is similar to that of Lemma \ref{ineq:SmallThomp2SmallOptError}, is deferred to Appendix \ref{sec:ApproxGradProof}.
\begin{remark}
    \label{remark:ApproxGrad}
    By Lemma \ref{lemma:ApproxGrad} and Theorem \ref{thm:PetzConv}, for any \( w \in \mathrm{relint}\left(\Delta_n\right) \) and $\alpha\in(1/2,1)$, to compute \( g(w) \) and \( \nabla g(w) \) with error smaller than \( \epsilon  \), it suffices to compute \( Q_{T+1} \)  using the iteration rule \eqref{alg:ProposedIter} for \( T = O\left(\log(1/\epsilon)\right) \), and then compute \( \hat{g}_w \) and \( \hat{\nabla} \) by setting \( Q = Q_{T+1}/\Tr\left[Q_{T+1}\right] \), as defined in Lemma \ref{lemma:ApproxGrad}.
\end{remark}

To identify a suitable gradient-based method, we bound the Hessian of $g(w)$ in Lemma \ref{lemma:HessianBd}, which establishes that $g(w)$ is $1$-smooth relative to the negative Shannon entropy $h(w)=\inner{w}{\log(w)}$ \citep{Bauschke2017,Lu2018,Birnbaum2011}, a generalization of standard $L$-smoothness in the convex optimization literature.
\begin{lemma}
    \label{lemma:HessianBd}
    Let $h(w)=\inner{w}{\log(w)}$.
    For any $\alpha\in\left[1/2,1\right)$, we have 
    \begin{equation}
        \nabla^2 g(w)
        \leq \nabla^2 h(w)
        =\mathrm{Diag}(w)^{-1},\forall w\in\mathrm{relint}\left(\Delta_n\right).
    \end{equation}
\end{lemma}
The proof of Lemma \ref{lemma:HessianBd} is deferred to Appendix \ref{sec:HessianBdProof}.
This relative smoothness of $g$, along with Lemma \ref{lemma:MirrorConv}, justifies the effectiveness of our chosen method \eqref{def:ProposedIterCapacity}, namely entropic mirror descent, which is an instance of the Bregman proximal method with reference function $h(w)=\inner{w}{\log(w)}$ \citep[Algorithm 1]{Lu2018}.
\begin{lemma}(\citep[Theorem 3.1]{Lu2018})
    \label{lemma:MirrorConv}
    Consider a generic convex optimization problem:
    \begin{equation}
        \hat{w}_{\star}
        \in \argmin_{\hat{w} \in \mathcal{X}} \hat{g}(\hat{w}),
    \end{equation}
    where \( \hat{g} \) is twice differentiable on \( \mathrm{relint}(\mathcal{X}) \).
    Suppose a convex, twice differentiable reference function \( \hat{h} \) is chosen such that
    \begin{align}
        \nabla^2 \hat{g}(\hat{w})
        \leq \nabla^2 \hat{h}(\hat{w}),\quad\forall \hat{w}\in\mathrm{relint}\left(\mathcal{X}\right).
    \end{align}
    Let the sequence \( \{ \hat{w}_t \}_{t \in \mathbb{N}} \) be recursively defined by
    \begin{equation}
        \hat{w}_{t+1}
        \in \argmin_{\hat{w} \in \mathcal{X}} \hat{g}(\hat{w}_t) + \langle \nabla \hat{g}(\hat{w}_t), \hat{w} - \hat{w}_t \rangle + B_{\hat{h}}(\hat{w} \| \hat{w}_t),
    \end{equation}
    with \( \hat{w}_1 \in \mathrm{relint}(\mathcal{X}) \).
    Then, for any \( T \in \mathbb{N} \), we have
    \begin{equation}
        \hat{g}(\hat{w}_{T+1}) - \hat{g}(\hat{w}_{\star})
        \leq \frac{B_{\hat{h}}(\hat{w}_{\star} \| \hat{w}_1)}{T}.
    \end{equation}
\end{lemma}

Below, we present our proposed iterative method for computing the Petz-Augustin capacity of order $\alpha\in(1/2,1)$, as defined in \eqref{def:Capacity}. 
The method performs entropic mirror descent as its outer loop, with each iteration invoking the iteration rule \eqref{alg:ProposedIter} as a subroutine.
\begin{itemize}
    \item Let $w_1=\mathbf{1}_n/n$ and define $h(w)=\inner{w}{\log(w)}$ for all $w\in\mathbb{R}_{+}^n$.
    \item For every $t\in\mathbb{N}$,  compute
    \begin{equation}
        \label{def:ProposedIterCapacity}
        w_{t+1}
        \in\argmin_{w\in\Delta_{n}}g(w_t)+\inner{\nabla g(w_t)}{w-w_t}+B_h(w\Vert w_t),
    \end{equation}
    or equivalently \citep[Eq. (11)]{He2024a},
    \begin{equation}
        w_{t+1}
        =\frac{w_t\odot\exp\left(-\nabla g(w_t)\right)}{\inner{w_t}{\exp\left(-\nabla g(w_t)\right)}},
    \end{equation}
    and output $\left(w_{t+1},g(w_{t+1})\right)$,
    where $\nabla g(w_t)$ and $g(w_{t+1})$ are computed via the iteration rule \eqref{alg:ProposedIter} (see Remark \ref{remark:ApproxGrad}).
\end{itemize}

\subsection{Convergence Analysis}
\label{subsec:CapacityConv}
\begin{theorem}
    \label{thm:ProposedCapacityConv}
    For any $\alpha\in(1/2,1)$, let $\Set{w_t}_{t\in\mathbb{N}}$ be the sequence of iterates generated by the iteration rule \eqref{def:ProposedIterCapacity}.
    Then, we have
    \begin{equation}
        g(w_{T+1})-g(w_{\star})
        \leq \frac{\log(n)}{T},
    \end{equation}
    where 
    \begin{equation}
        w_{\star}
        \in\argmin_{w\in\Delta_n}g(w).
    \end{equation}
\end{theorem}
\begin{proof}(Theorem \ref{thm:ProposedCapacityConv})
Let $w_t$ and $w_{\star}$ be defined as in Theorem \ref{thm:ProposedCapacityConv}, and let $h(w)=\inner{w}{\log(w)}$.
By Lemma \ref{lemma:HessianBd}, Lemma \ref{lemma:MirrorConv}, and the convexity of $g$ \citep[Prop. 5(b)]{Cheng2022}, we obtain \citep[Theorem 3.1]{Lu2018}
\begin{align}
    g(w_{T+1})-g(w_{\star})
    &\leq\frac{B_h(w_{\star}\Vert w_1)}{T}.
\end{align}
The proof then follows from
\begin{align}
    \frac{B_h(w_{\star}\Vert w_1)}{T}
    &=\frac{\inner{w_{\star}}{\log(w_{\star})-\log(w_1)}}{T}\nonumber\\
    &=\frac{\log(n)+\inner{w_{\star}}{\log(w_{\star})}}{T}\nonumber\\
    &\leq \frac{\log(n)}{T},
\end{align}
where the second equality uses the initialization $w_1=\mathbf{1}_n/n$. 
\end{proof}

\section{Adaptation to Fisher Market Equilibrium}
\label{sec:Fisher}
Unlike the computation of the Petz-Augustin capacity, where a direct application of the proposed iteration rule \eqref{alg:ProposedClassicalIter} suffices, applying the rule to determine equilibrium prices in Fisher markets requires nontrivial adaptations.
This section is devoted to adapting the proposed iteration rule \eqref{alg:ProposedClassicalIter} for computing equilibrium prices in  inhomogeneous Fisher markets where buyers have CES utilities with different elasticities and sellers update their prices asynchronously.
The remainder of this section is organized as follows:
\begin{itemize}
    \item Section \ref{subsec:FisherDef} introduces the problem formulation for computing equilibrium prices in such markets, establishes its correspondence with the classical Augustin mean when all buyers share a common elasticity in their utilities, and interprets the proposed iteration rule \eqref{alg:ProposedClassicalIter} as a t\^{a}tonnement process in this setting.
    \item Section \ref{subsec:FisherIter} proposes the adapted iteration rule \eqref{eq:AdaptedProposedIter}, which extends \eqref{alg:ProposedClassicalIter} to accommodate a more general setting.
    \item Section \ref{subsec:FisherIterConv} analyzes the convergence rate of the iteration rule \eqref{eq:AdaptedProposedIter}.
    \item Section \ref{subsec:TatonnementCmp} compares the iteration rule \eqref{eq:AdaptedProposedIter} with existing work.
\end{itemize}

\subsection{Introduction}
\label{subsec:FisherDef}
A Fisher market, under the assumption that all buyers have CES utilities, consists of $n$ buyers and $d$ sellers.
Each $i^{\text{th}}$ seller has the following responsibilities \citep{Cole2010}:
\begin{itemize}
    \item Holds a supply of one unit of the $i^{\text{th}}$ good, assumed to be infinitely divisible.
    \item Sets the price $p[i]$ per unit of the $i^{\text{th}}$ good.
\end{itemize}
The price vector $p\in\mathbb{R}_{+}^d$ is collectively determined by all sellers.

On the other hand, each $j^{\text{th}}$ buyer has:
\begin{itemize}
    \item A budget $w[j] > 0$. 
    Without loss of generality, we assume $\sum_{j=1}^n w[j]=1$.
    \item A CES utility function
    \begin{align}
        \label{Def:CESUtility}
        u_j : \mathbb{R}_{+}^d \to \mathbb{R}:
        x \mapsto \inner{a_j}{x^{\rho_j}}^{1/\rho_j},
    \end{align}
    where $a_j \in \mathbb{R}_{+}^d\setminus\Set{0}$ for all $j$, $\rho_j\in(-\infty,0)\cup(0,1)$\footnote{We exclude the extremal case $\rho_j=1$, as our analysis focuses on the range $(0,1)$, corresponding to the WGS regime.}, and $x[i]$ denotes the quantity (in units) of the $i^{\text{th}}$ good consumed. 
    Here, $\rho_j$ is the CES parameter, and the corresponding elasticity is $\sigma_j=1/(1-\rho_j)$.
    Without loss of generality, we assume $a_j\in\Delta_d$ for all $j$, and that $\sum_{j=1}^n a_j>0$.
    \item A demand vector that depends on the price vector $p$:
    \begin{equation}
        \label{def:FisherDemand}
        x_j(p)
        \in\argmax_{x\in\mathbb{R}_{+}^d,\inner{x}{p}\leq w[j]}u_j(x),
    \end{equation}
    where $x_j(p)$ is given in closed form by\footnote{We note that under the assumption \( \sum_{j=1}^n a_j > 0 \), the prices should remain positive to ensure that the demand remains bounded.}
    \begin{equation}
        x_j(p)
        =w[j]\frac{a_j^{1/(1-\rho_j)}\odot p^{-1/(1-\rho_j)} }{\inner{a_j^{1/(1-\rho_j)}}{p^{-\rho_j/(1-\rho_j)}}}.
    \end{equation}
\end{itemize}

We denote by 
\begin{equation}
    \label{def:FisherTotalDemand}
    x(p)
    \coloneqq \sum_{j=1}^n x_j(p)
    =\sum_{j=1}^n w[j]\frac{a_j^{1/(1-\rho_j)}\odot p^{-1/(1-\rho_j)} }{\inner{a_j^{1/(1-\rho_j)}}{p^{-\rho_j/(1-\rho_j)}}}
\end{equation} 
the total demand vector.

A price vector $p_{\star}$ is called the equilibrium price vector if  $x(p_{\star})=\mathbf{1}_{d}$, indicating that total demand matches the supply.

From the work of \citet[Lemma 3.3]{Cheung2020}, the equilibrium prices can equivalently be defined as the solution to the following convex program:
\begin{equation}
    \label{def:FisherEqPriceCvxOpt}
    p_{\star}
    \in\argmin_{p\in\mathbb{R}_{+}^d}\phi(p),\quad\phi(p)
    \coloneqq\sum_{j=1}^n w[j]\frac{1-\rho_j}{\rho_j}\log\left(\inner{a_j^{1/(1-\rho_j)}}{p^{-\rho_j/(1-\rho_j)}}\right)+\inner{\mathbf{1}_d}{p}.
\end{equation}
To establish the correspondence with the classical Augustin mean $q_{\star}$, we set $\rho_j=\rho=(\alpha-1)/\alpha$ for all $j$. 
Then, following the work of \citet[Lemma 5.2]{Wang2024}, we write
\begin{align}
    p_{\star}
    &\in\argmin_{p\in\mathbb{R}_{+}^d}\sum_{j=1}^n w[j]\frac{1-\rho_j}{\rho_j}\log\left(\inner{a_j^{1/(1-\rho_j)}}{p^{-\rho_j/(1-\rho_j)}}\right)+\inner{\mathbf{1}_d}{p}\nonumber\\
    &=\argmin_{p\in\Delta_{d}}\sum_{j=1}^n w[j]\frac{1}{\alpha-1}\log\left(\inner{a_j^{\alpha}}{p^{1-\alpha}}\right),
\end{align}
which recovers the optimization problem~\eqref{def:ClassicalAugMean} defining the classical Augustin mean $q_{\star}$.
Likewise, our proposed iteration rule \eqref{alg:ProposedClassicalIter} can be equivalently written as
\begin{itemize}
    \item Each $i^{\text{th}}$ seller initializes the price of the $i^{\text{th}}$ good to $p_1[i]>0$.
    \item For every $t\in\mathbb{N}$, all sellers synchronously update their prices according to
    \begin{equation}
        \label{eq:ProposedIterFisher}
        p_{t+1}
        =p_t\odot x(p_t)^{1-\rho}.
    \end{equation}
\end{itemize} 
Note that at each round $t$, each $i^{\text{th}}$ seller determines the price $p_{t+1}[i]$ using only the previous price $p_{t}[i]$, the corresponding demand $x(p_t)[i]$, and knowledge of the CES parameter $\rho$, indicating that the iteration rule \eqref{eq:ProposedIterFisher} can be interpreted as a t\^{a}tonnement process \citep{walras2014}.

\subsection{Adapted Iteration Rule}
\label{subsec:FisherIter}
To better capture actual market behavior, three additional issues must be addressed:
\begin{itemize}
    \item CES parameters $\rho_j$ vary across buyers.
    \item Sellers do not have precise knowledge of each $\rho_j$.
    \item Sellers update prices asynchronously.
\end{itemize}
When the CES utilities are in the WGS regime, where $\rho_j \in (0,1)$, we adapt the iteration rule \eqref{eq:ProposedIterFisher} to address these three issues as follows:
\begin{itemize}
    \item Each $i^{\text{th}}$ seller requires an upper bound $\hat{\rho}_i\in\left[\max_{j}\rho_j,1\right)$, rather than knowledge of all $\rho_j$.
    \item At each round $t$, sellers independently decide whether to update their price.
\end{itemize}
Observe that since $1-\rho>0$, the update becomes increasingly conservative as $\rho\to 1$; that is, $p_{t+1}[i]\approx p_t[i]$ when $1-\rho\to 0$.
Motivated by this, we treat $1-\rho$ in \eqref{eq:ProposedIterFisher} as a step size.
When each $i^{\text{th}}$ seller only has access to an upper bound $\hat{\rho}_i\in\left[\max_j\rho_j,1\right)$ on the buyers' CES parameters, it is natural to adopt the conservative step size $1-\hat{\rho}_i$.

We now present the adapted version of the iteration rule \eqref{eq:ProposedIterFisher} for computing the equilibrium prices $p_{\star}$.
\begin{itemize}
    \item Each $i^{\text{th}}$ seller initializes the price of the $i^{\text{th}}$ good to $p_1[i]>0$.
    \item For every $t\in\mathbb{N}$, select a non-empty subset $\mathcal{I}_t\subseteq\Set{1,2,\dots,d}$, representing the sellers who decide to update the prices of their goods.  
    For each index $i\in\mathcal{I}_t$, the price $p_{t+1}[i]$ is updated according to
    \begin{equation}
        \label{eq:AdaptedProposedIter}
        p_{t+1}[i]
        =p_t[i]x(p_t)[i]^{1-\hat{\rho}_i},
    \end{equation}
    and for each $k \notin \mathcal{I}_t$, the corresponding price remains unchanged:
    \begin{equation}
        p_{t+1}[k] = p_t[k].
    \end{equation}
\end{itemize}
The adapted iteration rule \eqref{eq:AdaptedProposedIter} can also be interpreted as a t\^{a}tonnement process, which reduces to the iteration rule \eqref{eq:ProposedIterFisher} when $\hat{\rho}_i=\rho$ for all $i$, and $\mathcal{I}_t=\Set{1,2,\dots,d}$ for all $t$.

\subsection{Convergence Analysis}
\label{subsec:FisherIterConv}
\begin{theorem}
\label{thm:FisherConv}
Let \( \rho_j \in (0,1) \) for all \( j \), and let \( \hat{\rho}_i \in \left[\max_j \rho_j,\, 1\right) \) for all \( i \).  
Let \( \{p_t\}_{t \in \mathbb{N}} \) be the sequence of iterates generated by the iteration rule \eqref{eq:AdaptedProposedIter}.
Furthermore, define $N(0)\coloneqq 0$ and
\begin{align}
    N(t+1)
    \coloneqq\min\Set{s>N(t)\mid
\bigcup_{k=N(t)+1}^{s}\mathcal I_k=\Set{1,\dots,d}},\quad\forall t\geq 0.
\end{align}

Then, we have
\begin{equation}
    \label{ineq:AdaptConvRate}
    d_{\mathrm{T}}\left(p_{\star}, p_{N(T)+1}\right) \leq \hat{\rho}^{T} \cdot d_{\mathrm{T}}\left(p_{\star}, p_1\right),
\end{equation}
where \( \hat{\rho} = \max_i \hat{\rho}_i \).

Moreover, since the equilibrium prices \( p_{\star} > 0 \) satisfying \( x(p_{\star}) = \mathbf{1}_d \) uniquely exist \citep[Lemma 1]{Zhang2011}, we have \( d_{\mathrm{T}}(p_{\star}, p_1) < \infty \).
\end{theorem}
We call the collection of iterations from the
$(N(t-1)+1)^{\text{th}}$ to the $N(t)^{\text{th}}$
iteration the $t^{\text{th}}$ \emph{epoch}.
The definition of $N(t)$ in Theorem~\ref{thm:FisherConv} ensures that the price of each good is updated at least once per epoch.
\begin{remark}
    Our analysis in Theorem~\ref{thm:FisherConv} can be extended to the range $\rho_j \in (-1,0)$.
    The CES parameters $\rho_j \in (-\infty, 0)$  correspond to the complementary regime.
    However, it is unclear whether our analysis can be extended to the entire complementary regime, and the operational meaning of restricting attention to the subrange $(-1, 0)$ is also unclear.
    Therefore, our analysis focuses exclusively on the WGS regime, that is, $\rho_j \in (0,1)$.
\end{remark}

\begin{proof}
The proof of Theorem \ref{thm:FisherConv} is similar to that of Theorem \ref{thm:PetzConv}, with the main difference being the handling of asynchronous updates using the technique developed by \citet{Dvijotham2022}.

The remainder of the proof proceeds as follows:
\begin{itemize}
    \item  We first show that after the $t^{\text{th}}$ round, for all $i\in\mathcal{I}_t$, the following coordinate-wise contractive property holds:
    \begin{equation}
        \log\left(\max\Set{\frac{p_{t+1}[i]}{p_{\star}[i]},\frac{p_{\star}[i]}{p_{t+1}[i]}}\right)
        \leq \hat{\rho}d_{\mathrm{T}}(p_{\star},p_t).
    \end{equation}
    For all $k\notin\mathcal{I}_t$, we have
    \begin{equation}
        \log\left(\max\Set{\frac{p_{t+1}[k]}{p_{\star}[k]},\frac{p_{\star}[k]}{p_{t+1}[k]}}\right)
        \leq d_{\mathrm{T}}(p_{\star},p_t).
    \end{equation}
    \item Based on the inequalities above, we establish the epoch-wise contractive property:
    \begin{equation}
        d_{\mathrm{T}}\left(p_{\star},p_{N(t+1)+1}\right)
        \leq\hat{\rho}d_{\mathrm{T}}\left(p_{\star},p_{N(t)+1}\right),
    \end{equation}
    which completes the proof.
\end{itemize}
To prove the coordinate-wise contractive property, fix any $i\in\mathcal{I}_t$.
From the iteration rule \eqref{eq:AdaptedProposedIter}, we have
\begin{equation}
    p_{t+1}[i]
    =p_t[i]x(p_t)[i]^{1-\hat{\rho}_i}
    =\left(\sum_{j=1}^n w[j]\frac{a_j[i]^{1/(1-\rho_j)}p_t[i]^{1/(1-\hat{\rho}_i)-1/(1-\rho_j)}}{\inner{a_j^{1/(1-\rho_j)}}{p_t^{-\rho_j/(1-\rho_j)}}}\right)^{1-\hat{\rho}_i},
\end{equation}
where the second equality follows from the definition of the total demand vector in \eqref{def:FisherTotalDemand}.
By Definition \ref{eq:Thompson}, it follows that
\begin{align}
    p_{t+1}[i]
    &\leq \left(\sum_{j=1}^n w[j]\frac{a_j[i]^{1/(1-\rho_j)}\left(\exp\left(d_{\mathrm{T}}(p_{\star},p_t)\right)p_{\star}[i]\right)^{1/(1-\hat{\rho}_i)-1/(1-\rho_j)}}{\inner{a_j^{1/(1-\rho_j)}}{\left(\exp\left(d_{\mathrm{T}}(p_{\star},p_t)\right)p_{\star}\right)^{-\rho_j/(1-\rho_j)}}}\right)^{1-\hat{\rho}_i}\nonumber\\
    &=\exp\left(\hat{\rho}_i d_{\mathrm{T}}(p_{\star},p_t)\right)\left(\sum_{j=1}^n w[j]\frac{a_j[i]^{1/(1-\rho_j)}p_{\star}[i]^{1/(1-\hat{\rho}_i)-1/(1-\rho_j)}}{\inner{a_j^{1/(1-\rho_j)}}{p_{\star}^{-\rho_j/(1-\rho_j)}}}\right)^{1-\hat{\rho}_i}\nonumber\\
    &=\exp\left(\hat{\rho}_i d_{\mathrm{T}}(p_{\star},p_t)\right)p_{\star}[i]x(p_{\star})[i]^{1-\hat{\rho}_i}\nonumber\\
    &=\exp\left(\hat{\rho}_i d_{\mathrm{T}}(p_{\star},p_t)\right)p_{\star}[i],
\end{align}
where the third equality follows from the fact that $x(p_{\star})=\mathbf{1}_d$.
A similar argument gives the lower bound:
\begin{align}
    p_{t+1}[i]
    &\geq \left(\sum_{j=1}^n w[j]\frac{a_j[i]^{1/(1-\rho_j)}\left(\exp\left(-d_{\mathrm{T}}(p_{\star},p_t)\right)p_{\star}[i]\right)^{1/(1-\hat{\rho}_i)-1/(1-\rho_j)}}{\inner{a_j^{1/(1-\rho_j)}}{\left(\exp\left(-d_{\mathrm{T}}(p_{\star},p_t)\right)p_{\star}\right)^{-\rho_j/(1-\rho_j)}}}\right)^{1-\hat{\rho}_i}\nonumber\\
    &=\exp\left(-\hat{\rho}_i d_{\mathrm{T}}(p_{\star},p_t)\right)p_{\star}[i].
\end{align}
Combining the upper and lower bounds, we obtain
\begin{equation}
    \log\left(\max\Set{\frac{p_{t+1}[i]}{p_{\star}[i]},\frac{p_{\star}[i]}{p_{t+1}[i]}}\right)
    \leq \hat{\rho}_id_{\mathrm{T}}(p_{\star},p_t)
    \leq \hat{\rho}d_{\mathrm{T}}(p_{\star},p_t).
\end{equation}
On the other hand, for any $k\notin\mathcal{I}_t$, the price remains unchanged, i.e., $p_{t+1}[k]=p_t[k]$, so we have
\begin{align}
    \log\left(\max\Set{\frac{p_{t+1}[k]}{p_{\star}[k]},\frac{p_{\star}[k]}{p_{t+1}[k]}}\right)
    &=\log\left(\max\Set{\frac{p_{t}[k]}{p_{\star}[k]},\frac{p_{\star}[k]}{p_{t}[k]}}\right)\nonumber\\
    &\leq\max_{1\leq l\leq d}\log\left(\max\Set{\frac{p_{t}[l]}{p_{\star}[l]},\frac{p_{\star}[l]}{p_{t}[l]}}\right)\nonumber\\
    &=d_{\mathrm{T}}(p_{\star},p_t).
\end{align}
This completes the proof of the coordinate-wise contractive property.

To prove the epoch-wise contractive property, we first show that the Thompson metric is non-increasing:
\begin{align}
    d_{\mathrm{T}}(p_{\star},p_{t+1})
    &=\max_{1\leq l\leq d}\log\left(\max\Set{\frac{p_{t+1}[l]}{p_{\star}[l]},\frac{p_{\star}[l]}{p_{t+1}[l]}}\right)\nonumber\\
    &=\max\Set{\max_{i\in\mathcal{I}_t}\log\left(\max\Set{\frac{p_{t+1}[i]}{p_{\star}[i]},\frac{p_{\star}[i]}{p_{t+1}[i]}}\right),\max_{k\notin\mathcal{I}_t}\log\left(\max\Set{\frac{p_{t+1}[k]}{p_{\star}[k]},\frac{p_{\star}[k]}{p_{t+1}[k]}}\right)}\nonumber\\
    &\leq\max\Set{\hat{\rho} d_{\mathrm{T}}(p_{\star},p_t),d_{\mathrm{T}}(p_{\star},p_t)}\nonumber\\
    &\leq d_{\mathrm{T}}(p_{\star},p_t),
\end{align}
where the first inequality follows from the coordinate-wise contractive property.
Next, by the definition of $N(t)$, for each $i$, there exists some $\tau\in\Set{N(t)+1,N(t)+2,\dots,N(t+1)}$ such that $i\in\mathcal{I}_{\tau}$.
For each $i$, we define 
\begin{equation}
    \tau_t[i]
    \coloneqq\max\Set{\tau\in\Set{N(t)+1,N(t)+2,\dots,N(t+1)}\mid i\in \mathcal{I}_{\tau}}.
\end{equation}
It then follows from the coordinate-wise contractive property and the monotonicity of the Thompson metric that
\begin{align}
    d_{\mathrm{T}}\left(p_{\star},p_{N(t+1)+1}\right)
    &=\max_{1\leq i\leq d}\log\left(\max\Set{\frac{p_{N(t+1)+1}[i]}{p_{\star}[i]},\frac{p_{\star}[i]}{p_{N(t+1)+1}[i]}}\right)\nonumber\\
    &= \max_{1\leq i\leq d}\log\left(\max\Set{\frac{p_{\tau_t[i]+1}[i]}{p_{\star}[i]},\frac{p_{\star}[i]}{p_{\tau_t[i]+1}[i]}}\right)\nonumber\\
    &\leq \max_{1\leq i\leq d}\hat{\rho} d_{\mathrm{T}}\left(p_{\star},p_{\tau_t[i]}\right)\nonumber\\
    &\leq \max_{1\leq i\leq d}\hat{\rho} d_{\mathrm{T}}\left(p_{\star},p_{N(t)+1}\right)\nonumber\\
    &=\hat{\rho} d_{\mathrm{T}}\left(p_{\star},p_{N(t)+1}\right).\nonumber
\end{align}
This concludes the proof.
\end{proof}
\subsection{Comparison with Existing Work}
\label{subsec:TatonnementCmp}



We compare our proposed iteration rule \eqref{eq:AdaptedProposedIter} with that of \citet{Cole2008}, which, to the best of our knowledge, is the only t\^{a}tonnement-type algorithm for Fisher markets with WGS CES utilities that both allows asynchronous price updates and guarantees linear convergence of prices to equilibrium.
Although other t\^{a}tonnement-type algorithms proposed by \citet{Bei2019,Cheung2020,Cheung2018ESA} are also applicable in this setting, their convergence guarantees are stated using error measures that do not quantify the distance between prices and the equilibrium, making a direct comparison with our results infeasible.
In contrast, while \citet{Cole2008} adopt a different error measure to quantify the distance between a price vector $p$ and the equilibrium price vector $p_{\star}$, denoted $d_{\mathrm{CF}}\left(p_{\star},p\right)$, it simplifies to $\max_{i}\abs{1-p[i]/p_{\star}[i]}$ when $d_{\mathrm{T}}\left(p_{\star},p\right)<\log(3)$.
By Lemma \ref{lemma:MetricComparable}, this is comparable to the Thompson metric we use.

For comparison, we first focus on the case in which each $j^{\text{th}}$ buyer has a CES utility with parameter $\rho_j\in(0,1)$, and each $i^{\text{th}}$ seller knows an upper bound $\hat{\rho}_i\in\left[\max_j\rho_j,1\right)$.
For simplicity, assume sellers update their prices synchronously.
In this setting, their algorithm computes prices that are $\epsilon$-close to the equilibrium in 
\begin{equation}
    O\left(\frac{1+\hat{\rho}}{1-\hat{\rho}}\log\left(\frac{1}{\epsilon}\right)\right)
\end{equation}
iterations \citep[Theorem 3.1]{Cole2008}, where $\hat{\rho}=\max_{i}\hat{\rho}_i$.
In contrast, our iteration rule \eqref{eq:AdaptedProposedIter} requires 
\begin{equation}
    O\left(\frac{1}{\log\left(\frac{1}{\hat{\rho}}\right)}\log\left(\frac{1}{\epsilon}\right)\right)
\end{equation}
iterations (Theorem \ref{thm:FisherConv}), yielding a better dependence on $\hat{\rho}$.
In particular, 
\begin{equation}
    \frac{1}{\log(\frac{1}{\hat{\rho}})}
    \leq\frac{1+\hat{\rho}}{1-\hat{\rho}};
\end{equation}
moreover, the former tends to $0$ as $\hat{\rho}$ approaches $0$, while the latter converges to $1$.
However, in the general case, \citet{Cole2008} do not require sellers to have any knowledge of the CES parameters $\rho_j$, whereas our method requires knowledge of $\hat{\rho}_i$.
Removing this requirement is left as a future direction.

\section{Numerical Results} \label{sec:numerics}
We implement our proposed iteration rule \eqref{alg:ProposedIter} to compute the Petz-Augustin mean \eqref{def:PetzAugMean} for $\alpha$ in $\set{0.8,1.5,3,5}$. 
The source code is available on GitHub\footnote{\url{https://github.com/chunnengchu/PetzAugustin/}}.
We set the number of Petz-R\'{e}nyi divergences to $n=2^{5}$ and the dimension of the quantum states to $d=2^{7}$, generating the quantum states $ A_j $ using the  
\texttt{rand\_dm}
function from the 
Python 
package QuTiP \citep{Johansson2012}.
Since the exact optimizer is unknown, we compute the approximate optimization error and the approximate iterate error instead. 
The former is defined as $F\left(Q_t/\Tr\left[Q_t\right]\right)-F\left(\hat{Q}_{\star}\right)$, and the latter  as $d_{\mathrm{T}}\left(\hat{Q}_{\star}^{1-\alpha},\left(Q_t/\Tr[Q_t]\right)^{1-\alpha}\right)$, where $\hat{Q}_\star$ denotes the last iterate after $60$ iterations of the iteration rule \eqref{alg:ProposedIter}.


In Figure \ref{fig1}, we observe linear convergence rates for the optimization error and the iterate error when $\alpha>0.5$.
Notably, Theorem \ref{thm:PetzConv} establishes that the exponent of the linear convergence rate is bounded above by $\left|1 - \frac{1}{\alpha}\right|$.
Consistent with this result, Figure \ref{fig1} demonstrates a similar relationship between $\alpha$ and the
empirical 
convergence rate.


\begin{figure}[H]
    \centering

    \includegraphics[width=0.8\textwidth]{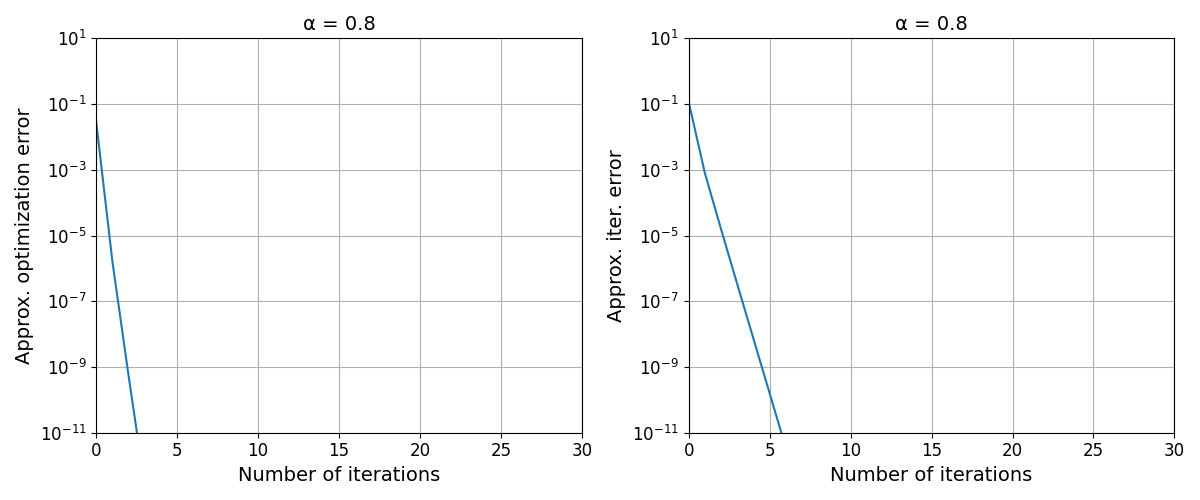}
    
    \includegraphics[width=0.8\textwidth]{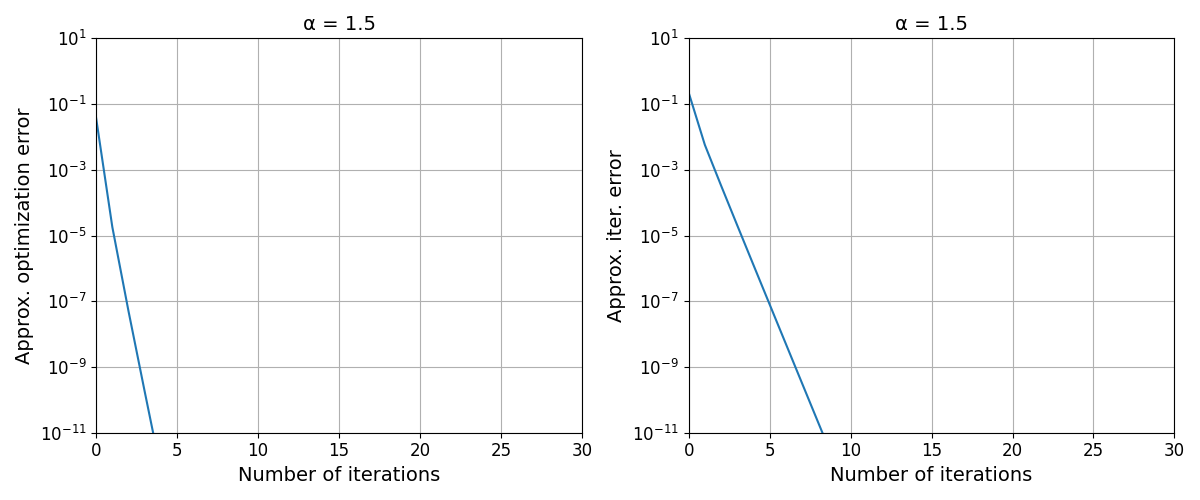}
    
    \includegraphics[width=0.8\textwidth]{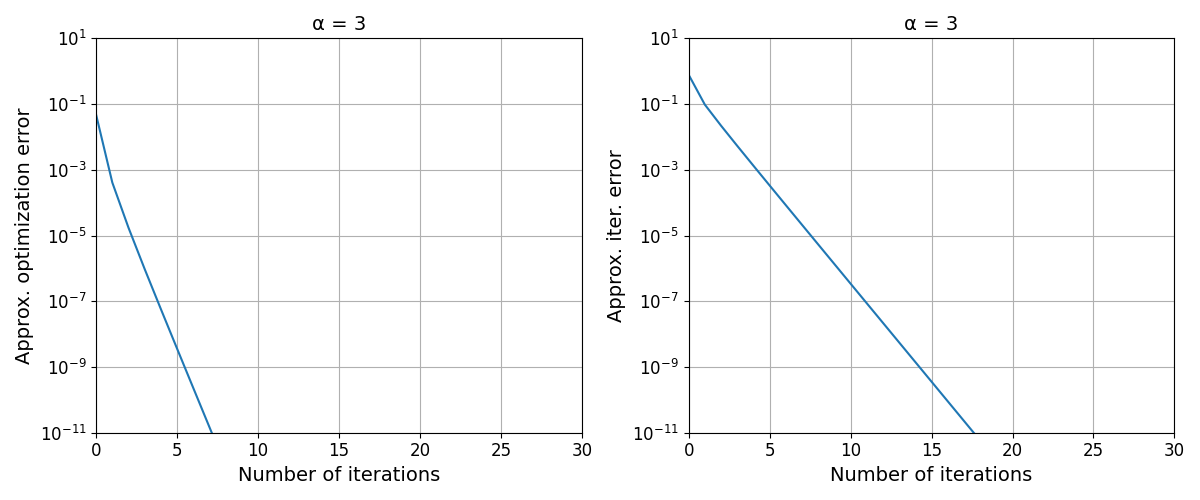}
    
    \includegraphics[width=0.8\textwidth]{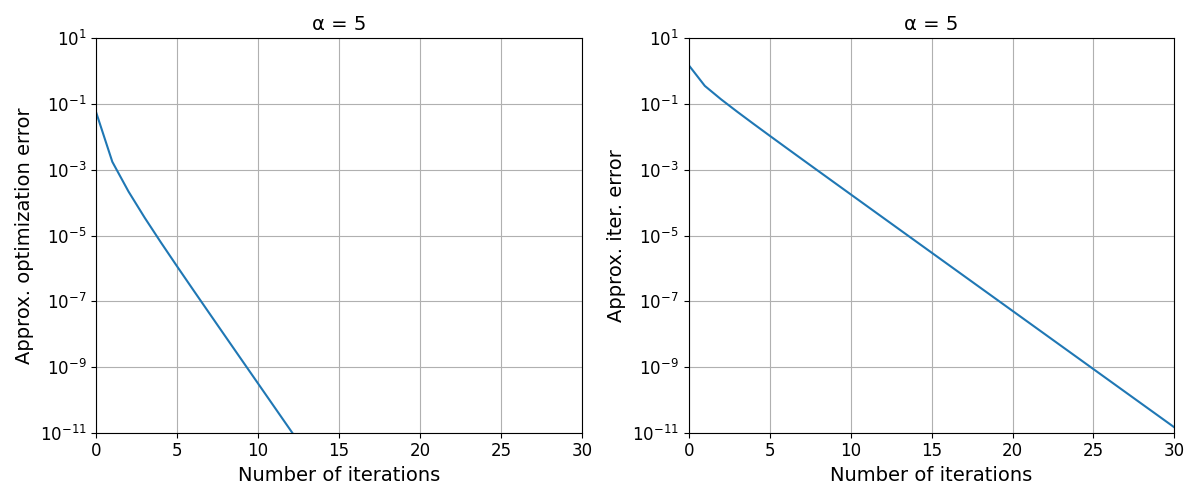}

    \caption{Approximate optimization error and iterate error versus the number of iterations.}
    \label{fig1}
\end{figure}

    





\section*{AI Disclosure}
The authors developed the algorithm for computing the Petz-Augustin mean, together with its analysis and applications to computing the Petz-Augustin capacity and equilibrium prices in Fisher markets.
ChatGPT provided the counterexample in Section~\ref{subsubsec:FailureAlphaSmall}, which shows that our iteration rule may fail to converge for $\alpha\in(0,1/2]$, and identified an error in the convergence analysis of \citet{Dvijotham2022} for the best response dynamic, as discussed in Appendix~\ref{appendix:DvijothamCounterexample}.

\section*{Acknowledgments}
This work was supported by the National Science and Technology Council of Taiwan under Grant Nos. NSTC 112-2628-E-002-019-MY3 and NSTC 115-2221-E-002-208, by the research project “Geometry of Quantum Learning and Optimization” at National Taiwan University under Grant No. NTU-CC-115L893706, and by the Young Theoretical Scientist Award of the Mathematics Division of the National Center for Theoretical Sciences of Taiwan.

\bibliography{refs}

\appendix
\section{Notations}
For any $v\in\mathbb{C}^d$, we denote its conjugate transpose by $v^*$.
We denote by $\mathcal{L}(\mathcal{B}_d,\mathcal{B}_d)$ the set of all linear operators from $\mathcal{B}_d$ to $\mathcal{B}_d$.
We denote by $\mathcal{S}_{d}$ the set of all self-adjoint linear operators in $\mathcal{L}(\mathcal{B}_d,\mathcal{B}_d)$.
For any $L_1,L_2\in\mathcal{S}_d$, we write $L_1\leq L_2$ if and only if $L_2-L_1$ is positive semi-definite.
Similarly, we write $L_1<L_2$ if and only if $L_2-L_1$ is positive definite.
We define $\mathcal{S}_{d,+}$ and $\mathcal{S}_{d,++}$ as the nonnegative cone and the positive cone in $\mathcal{S}_{d}$, respectively.
Specifically,
\begin{equation}
    \mathcal{S}_{d,+}
    \coloneqq \Set{L\in\mathcal{S}_d\mid L\geq 0},
\end{equation}
and
\begin{equation}
    \mathcal{S}_{d,++}
    \coloneqq \Set{L\in\mathcal{S}_{d}\mid L>0}.
\end{equation}
For any Fr\'{e}chet differentiable function $\phi:\mathcal{B}_d\to\mathcal{B}_d$, we denote by $\mathrm{D}_U\phi(U)\in\mathcal{L}(\mathcal{B}_d,\mathcal{B}_d)$ its Fr\'{e}chet derivative with respect to $U$.
For any finite-dimensional real Hilbert space $\mathcal{H}$ and any Fr\'{e}chet differentiable function $\psi:\mathbb{R}^n\to\mathcal{H}$, we denote by $\frac{\partial}{\partial v[i]}\psi(v)\in\mathcal{H}$ its partial derivative with respect to the $i^{\text{th}}$ coordinate of $v$.

\section{Proof of Lemma \ref{lemma:Danskin}}
\label{sec:DanskinProof}
This section is devoted to the proof of Lemma \ref{lemma:Danskin}.
We first establish in Lemma \ref{lemma:MeanDiff} that the Petz-Augustin mean $Q_{\star}(w)$ defined in \eqref{def:Capacity} is Fr\'{e}chet differentiable with respect to $w$, relying on Lemma \ref{lemma:PowerInfDiff} and the implicit function theorem (Lemma \ref{lemma:ImpFunc}).
As a consequence, the function $g$ is also differentiable.
Finally, combining this with the convexity of $g$, we conclude the proof by showing that, for each $w$, the vector $v(w)$ defined by
\begin{equation}
    v(w)[j]
    = \frac{\log\left(\Tr\left[A_j^{\alpha} Q_{\star}(w)^{1-\alpha}\right]\right)}{1-\alpha},
    \quad \forall j \in \Set{1,2,\dots,n},
\end{equation}
serves as a subgradient of $g(w)$.

\begin{lemma}[{Implicit Function Theorem \citep[Theorem 9.28]{Rudin1976}}]
    \label{lemma:ImpFunc}
    Let $\phi:\mathbb{R}^{n+m}\to\mathbb{R}^{n}$ be continuously differentiable on an open set $\mathcal{U}\subset\mathbb{R}^{n+m}$.
    Suppose  $\left(\bar{a},\bar{b}\right)\in\mathcal{U}$ satisfies $\phi(\bar{a},\bar{b})=0$ and the Jacobian matrix $\nabla_a \phi(a,b)\in\mathbb{R}^{n\times n}$ is invertible at $(a,b)=\left(\bar{a},\bar{b}\right)$.
    Then, there exist open sets $\mathcal{V}\subset\mathbb{R}^{n+m}$ and $\mathcal{W}\subset\mathbb{R}^{m}$ with $\bar{b}\in\mathcal{W}$ and $(\bar{a},\bar{b})\in\mathcal{V}$, and a unique continuously differentiable function $\psi:\mathcal{W}\to\mathbb{R}^{n}$ such that $\phi(\psi(b),b)=0$ and $(\psi(b),b)\in\mathcal{V}$ for all $b\in\mathcal{W}$, with $\psi(\bar{b})=\bar{a}$.
\end{lemma}

\begin{lemma}[{\citet[Theorem 3.33]{Hiai2014}}]
    \label{lemma:PowerInfDiff}
    For any $r\in\mathbb{R}$, the function $\phi(U)=U^{r}$ is infinitely Fr\'{e}chet differentiable on $\mathcal{B}_{d,++}$, and its Fr\'{e}chet derivative at $U=Q\in\mathcal{B}_{d,++}$ is given by
    \begin{equation}
        \mathrm{D}_Q \phi(Q)(V)
        =\sum_{1\leq i,j\leq d}h_{r}^{[1]}(\lambda_i,\lambda_j)u_iu_i^{*}Vu_ju_j^{*},\quad \forall V\in\mathcal{B}_d,
    \end{equation}
    where $Q=\sum_{i=1}^d\lambda_i u_iu_i^{*}$ is the eigendecomposition of $Q$, and
    \begin{equation}
        h_r^{[1]}(a,b)
        \coloneqq
        \begin{cases}
            \frac{a^r-b^r}{a-b}, & \text{if $a\neq b$.}\\
            ra^{r-1}, & \text{otherwise.}
        \end{cases}\quad \forall a,b\in\mathbb{R}_{++}.
    \end{equation}
\end{lemma}

    

\begin{lemma}
    \label{lemma:MeanDiff}
    The Petz-Augustin mean $Q_{\star}(w)$ for $\alpha\in(0,1)$, as defined in \eqref{def:Capacity}, is Fr\'{e}chet differentiable with respect to $w\in\mathbb{R}_{++}^n$. 
\end{lemma}
\begin{proof}
Let $\alpha\in(0,1)$.
We note that the Petz-Augustin mean $Q_{\star}(w)$ satisfies the following fixed-point property (Lemma \ref{eq:FixPtIsMin}):
\begin{equation}
    Q_{\star}(w)
    =\left(\sum_{j=1}^n\frac{w[j]}{\inner{\mathbf{1}_n}{w}}\frac{A_j^{\alpha}}{\Tr\left[A_j^{\alpha}Q_{\star}(w)^{1-\alpha}\right]}\right)^{1/\alpha}=\phi_2(\phi_1(w),w),
\end{equation}
where 
\begin{equation}
    \phi_1:\mathbb{R}_{++}^n\to\mathbb{R}^n:w\mapsto \left(\frac{\log\left(\Tr\left[A_1^{\alpha}Q_{\star}(w)^{1-\alpha}\right]\right)}{\alpha-1},\dots,\frac{\log\left(\Tr\left[A_n^{\alpha}Q_{\star}(w)^{1-\alpha}\right]\right)}{\alpha-1}\right),
\end{equation}
and
\begin{equation}
    \phi_2:\mathbb{R}^n\times\mathbb{R}_{++}^n\to\mathcal{B}_{d,++}:(v,w)\mapsto \left(\sum_{j=1}^n \frac{w[j]}{\inner{\mathbf{1}_n}{w}}\exp\left((1-\alpha)v[j]\right)A_j^{\alpha}\right)^{1/\alpha}.
\end{equation}
By Lemma \ref{lemma:PowerInfDiff}, the function $U\mapsto U^{1/\alpha}$ is infinitely Fr\'{e}chet differentiable on $\mathcal{B}_{d,++}$, implying that $\phi_2$ is also infinitely Fr\'{e}chet differentiable on $\mathbb{R}^{n}\times\mathbb{R}_{++}^{n}$.
Then, it remains to show that $\phi_1$ is differentiable on $\mathbb{R}_{++}^n$.

We start by defining an auxiliary function $\phi:\mathbb{R}^{n}\times\mathbb{R}_{++}^{n}\to\mathbb{R}^{n}$:
\begin{equation}
    \phi(v,w)[j]
    \coloneqq v[j]-\frac{\log\left(\Tr\left[A_j^{\alpha}\phi_2(v,w)^{1-\alpha}\right]\right)}{\alpha-1},\quad \forall j\in\Set{1,2,\dots,n},\;(v,w)\in\mathbb{R}^{n}\times\mathbb{R}_{++}^{n}.
\end{equation}

For each $\bar{w}\in\mathbb{R}_{++}^n$, we have $\phi(\bar{v},\bar{w})=0$, where $\bar{v}=\phi_1(\bar{w})$.
We aim to prove that $\bar{v}$ is differentiable with respect to $\bar{w}$, using the implicit function theorem (Lemma \ref{lemma:ImpFunc}).
To this end, we first note that, by Lemma \ref{lemma:PowerInfDiff}, the function 
\(U \mapsto U^{1-\alpha}\) is infinitely Fr\'{e}chet differentiable. 
Combined with the fact that \(\phi_2\) is also infinitely Fr\'{e}chet differentiable, this implies that the function \(\phi\) is continuously differentiable on 
\(\mathbb{R}^n \times \mathbb{R}_{++}^n\).
On the other hand, we compute the Jacobian matrix $\nabla_{v}\phi(v,w)$ at $(v,w)=(\bar{v},\bar{w})$ as follows:
\begin{align}
    &\nabla_{v}\phi(v,w)[j][i]\nonumber\\
    &\quad=I_n[j][i]+\frac{1}{1-\alpha}\frac{\Tr\left[A_j^{\alpha}\left(\frac{\partial}{\partial v[i]}\phi_2(v,w)^{1-\alpha}\right)\right]}{\Tr\left[A_j^{\alpha}\phi_2(v,w)^{1-\alpha}\right]}\nonumber\\
    &\quad=I_n[j][i]+\frac{1}{1-\alpha}\frac{\Tr\left[A_j^{\alpha}\left(\left(\mathrm{D}_{U}U^{(1-\alpha)/\alpha}\right)\left(\frac{\partial}{\partial v[i]}\left(\sum_{k=1}^n \frac{w[k]}{\inner{\mathbf{1}_n}{w}}\exp\left((1-\alpha)v[k]\right)A_k^{\alpha}\right)\right)\right)\right]}{\Tr\left[A_j^{\alpha}\phi_2(v,w)^{1-\alpha}\right]}\nonumber\\
    &\quad=I_n[j][i]+\frac{w[i]}{\inner{\mathbf{1}_n}{w}}\Tr\left[\frac{A_j^{\alpha}}{\Tr\left[A_j^{\alpha}Q_{\star}(w)^{1-\alpha}\right]}\left(\left(\mathrm{D}_{U}U^{(1-\alpha)/\alpha}\right)\left(\frac{A_i^{\alpha}}{\Tr\left[A_i^{\alpha}Q_{\star}(w)^{1-\alpha}\right]}\right)\right)\right],\nonumber
\end{align}
for all $i,j\in\Set{1,2,\dots,n}$, where 
\begin{equation}
    U
    =\phi_2(v,w)^{\alpha}
    =\sum_{k=1}^{n}\frac{\frac{w[k]}{\inner{w}{\mathbf{1}_n}}A_k^{\alpha}}{\Tr\left[{A_k^{\alpha}}{Q_{\star}(w)^{1-\alpha}}\right]}
    =Q_{\star}(w)^{\alpha},
\end{equation}
and the equalities follow from the chain rule.

By Lemma \ref{lemma:PowerInfDiff}, we obtain
\begin{align}
    \nabla_{v}\phi(v,w)
    =\left(\mathrm{Diag}\left(\frac{\bar{w}}{\inner{\mathbf{1}_n}{\bar{w}}}\right)^{-1}+G\right)\mathrm{Diag}\left(\frac{\bar{w}}{\inner{\mathbf{1}_n}{\bar{w}}}\right),\nonumber
\end{align}
where
\begin{align}
    G[j][i]
    =\Tr\left[\sum_{1\leq k,l\leq d}h_{(1-\alpha)/\alpha}^{[1]}\left(\lambda_k^{\alpha},\lambda_l^{\alpha}\right)u_k u_k^{*}\frac{A_j^{\alpha}}{\Tr\left[A_j^{\alpha}Q_{\star}(\bar{w})^{1-\alpha}\right]}u_lu_l^{*}\frac{A_i^{\alpha}}{\Tr\left[A_i^{\alpha}Q_{\star}(\bar{w})^{1-\alpha}\right]}\right],
\end{align}
and $Q_{\star}(\bar{w})=\sum_{k=1}^d\lambda_k u_ku_k^{*}$ is the eigendecomposition of $Q_{\star}(\bar{w})$.

Moreover, for any $z\in\mathbb{R}^n$, we have
\begin{align}
    \inner{Gz}{z}
    &=\Tr\left[\sum_{1\leq k,l\leq d}h_{(1-\alpha)/\alpha}^{[1]}\left(\lambda_k^{\alpha},\lambda_l^{\alpha}\right) u_ku_k^{*}\left(\sum_{j=1}^n\frac{z[j]A_j^{\alpha}}{\Tr\left[A_j^{\alpha}Q_{\star}(\bar{w})^{1-\alpha}\right]}\right)u_lu_l^{*}\left(\sum_{i=1}^n\frac{z[i]A_i^{\alpha}}{\Tr\left[A_i^{\alpha}Q_{\star}(\bar{w})^{1-\alpha}\right]}\right)\right]\nonumber\\
    &=\sum_{1\leq k,l\leq d}h_{(1-\alpha)/\alpha}^{[1]}\left(\lambda_k^{\alpha},\lambda_l^{\alpha}\right)\abs{ u_k^{*}\left(\sum_{j=1}^n\frac{z[j]A_j^{\alpha}}{\Tr\left[A_j^{\alpha}Q_{\star}(\bar{w})^{1-\alpha}\right]}\right)u_l}^2\nonumber\\
    &\geq 0.\nonumber
\end{align}
Thus, the matrix $G\in\mathbb{R}^{n\times n}$ is positive semi-definite, and it follows that the Jacobian matrix $\nabla_v\phi(v,w)$ at $(v,w)=(\bar{v},\bar{w})$ is invertible.

By the implicit function theorem (Lemma \ref{lemma:ImpFunc}), there exists an open set $\mathcal{W}_{\bar{w}}\subset\mathbb{R}_{++}^n$ and a continuously differentiable function $\psi_{\bar{w}}:\mathcal{W}_{\bar{w}}\to\mathbb{R}^{n}$ such that $\phi(\psi_{\bar{w}}(w),w)=0$ for all $w\in\mathcal{W}_{\bar{w}}$ and $\psi_{\bar{w}}(\bar{w})=\bar{v}$.
It remains to prove that 
\begin{equation}
    \psi_{\bar{w}}(w)=\phi_1(w),\quad \forall\bar{w}\in\mathbb{R}_{++}^n,w\in\mathcal{W}_{\bar{w}}.\nonumber
\end{equation} 

From the equality $\phi(\psi_{\bar{w}}(w),w)=0$, we obtain 
\begin{align}
    \psi_{\bar{w}}(w)[j]
    =\frac{\log\left(\Tr\left[A_j^{\alpha}\phi_2\left(\psi_{\bar{w}}(w),w\right)^{1-\alpha}\right]\right)}{\alpha-1}.
\end{align}
It follows that 
\begin{align}
    \phi_2\left(\psi_{\bar{w}}(w),w\right)
    =\left(\sum_{j=1}^n\frac{w[j]}{\inner{\mathbf{1}_n}{w}}\frac{A_j^{\alpha}}{\Tr\left[A_j^{\alpha}\phi_2\left(\psi_{\bar{w}}(w),w\right)^{1-\alpha}\right]}\right)^{1/\alpha},
\end{align}
which, by Lemma \ref{eq:FixPtIsMin}, implies that 
\begin{align}
    \phi_2\left(\psi_{\bar{w}}(w),w\right)
    =Q_{\star}(w).
\end{align}
Hence, we have
\begin{align}
    \psi_{\bar{w}}(w)[j]
    &=\frac{\log\left(\Tr\left[A_j^{\alpha}\phi_2\left(\psi_{\bar{w}}(w),w\right)^{1-\alpha}\right]\right)}{\alpha-1}\nonumber\\
    &=\frac{\log\left(\Tr\left[A_j^{\alpha}Q_{\star}(w)^{1-\alpha}\right]\right)}{\alpha-1}\nonumber\\
    &=\phi_1(w)[j],\quad \forall j\in\Set{1,2,\dots,n}.\nonumber
\end{align}
This concludes the proof.
\end{proof}

\begin{proof}(Lemma \ref{lemma:Danskin})
Let $w$ and $\alpha$ be defined as in Lemma \ref{lemma:Danskin}, and let $g$ be as in \eqref{def:Capacity}.
By Lemma \ref{lemma:MeanDiff}, the function $g$ is differentiable on $\mathbb{R}_{++}^n$.
Moreover, since $g$ is convex \citep[Prop. 5(b)]{Cheng2022}, it follows that there exists a unique subgradient $v(w)\in\mathbb{R}^n$ such that
\begin{align}
    g(z)
    \geq g(w)+\inner{v(w)}{z-w},\quad\forall z\in\mathbb{R}_{++}^n,
\end{align}
and this subgradient $v(w)$ coincides with $\nabla g(w)$.
On the other hand, by the definition of $g$, we have
\begin{align}
    g(z)
    &=-\sum_{j=1}^n z[j]\frac{\log\left(\Tr\left[A_j^{\alpha}Q_{\star}(z)^{1-\alpha}\right]\right)}{\alpha-1}\nonumber\\
    &=-\min_{Q\in\mathcal{D}_d}\sum_{j=1}^n z[j]\frac{\log\left(\Tr\left[A_j^{\alpha}Q^{1-\alpha}\right]\right)}{\alpha-1}\nonumber\\
    &\geq -\sum_{j=1}^n z[j]\frac{\log\left(\Tr\left[A_j^{\alpha}Q_{\star}(w)^{1-\alpha}\right]\right)}{\alpha-1}\nonumber\\
    &=g(w)+\sum_{j=1}^n\frac{\log\left(\Tr\left[A_j^{\alpha}Q_{\star}(w)^{1-\alpha}\right]\right)}{1-\alpha}(z[j]-w[j]),\quad\forall z\in\mathbb{R}_{++}^n.
\end{align}
This concludes the proof.    
\end{proof}

\section{Proof of Lemma \ref{lemma:ApproxGrad}}
\label{sec:ApproxGradProof}
\begin{proof}(Lemma \ref{lemma:ApproxGrad})
    Let $w$, $\alpha$, $Q$, $\hat{g}_w$, and $\hat{\nabla}$ be defined as in Lemma \ref{lemma:ApproxGrad}.
    The inequality 
    \begin{equation}
        \abs{\hat{g}_w-g(w)}
        \leq\abs{\frac{1}{1-\alpha}}d_{\mathrm{T}}\left(Q_{\star}(w)^{1-\alpha},Q^{1-\alpha}\right)
    \end{equation}
    follows directly from Lemma \ref{ineq:SmallThomp2SmallOptError}.
    By Lemma \ref{lemma:Danskin}, it remains to prove that
    \begin{align}
        \abs{D_{\alpha}\left(A_j\Vert Q\right)-D_{\alpha}\left(A_j\Vert Q_{\star}(w)\right)}
        &=\abs{\frac{1}{\alpha-1}}\abs{\log\left(\Tr\left[A_j^{\alpha}Q^{1-\alpha}\right]\right)-\log\left(\Tr\left[A_j^{\alpha}Q_{\star}(w)^{1-\alpha}\right]\right)}\nonumber\\
        &\leq\abs{\frac{1}{1-\alpha}}d_{\mathrm{T}}\left(Q_{\star}(w)^{1-\alpha},Q^{1-\alpha}\right),\quad\forall j\in\Set{1,2,\dots,n}.\nonumber 
    \end{align}
    It suffices to prove that
    \begin{align}
        \log\left(\Tr\left[A_j^{\alpha}U^{1-\alpha}\right]\right)-\log\left(\Tr\left[A_j^{\alpha}V^{1-\alpha}\right]\right)
        \leq d_{\mathrm{T}}\left(V^{1-\alpha},U^{1-\alpha}\right),\quad\forall U,V\in\mathcal{B}_{d,++}.
    \end{align}
    By Definition \ref{eq:Thompson}, we write
    \begin{align}
        &\log\left(\Tr\left[A_j^{\alpha}U^{1-\alpha}\right]\right)-\log\left(\Tr\left[A_j^{\alpha}V^{1-\alpha}\right]\right)\nonumber\\
        &\quad\leq\log\left(\Tr\left[A_j^{\alpha}\left(\exp\left(d_{\mathrm{T}}\left(V^{1-\alpha},U^{1-\alpha}\right)\right)V^{1-\alpha}\right)\right]\right)-\log\left(\Tr\left[A_j^{\alpha}V^{1-\alpha}\right]\right)\nonumber\\
        &\quad=d_{\mathrm{T}}\left(V^{1-\alpha},U^{1-\alpha}\right).
    \end{align}
    This concludes the proof.
\end{proof}

\section{Proof of Lemma \ref{lemma:HessianBd}}
\label{sec:HessianBdProof}
This section is devoted to proving Lemma \ref{lemma:HessianBd}.
We begin by computing the Hessian $\nabla^2 g(w)$ in Lemma \ref{lemma:HessianForm}.
Next, we apply Lemma \ref{lemma:ad-hoc} and Lemma \ref{lemma:LiebAndo} to prove Lemma \ref{lemma:L1geqL2}.
Finally, we combine Lemma \ref{lemma:HessianForm} and Lemma \ref{lemma:L1geqL2} to establish Lemma \ref{lemma:HessianBd}.

\begin{lemma}
    \label{lemma:HessianForm}
    Let $g$ and $Q_{\star}(w)$ be defined as in \eqref{def:Capacity}.
    For any $\alpha\in(0,1)$ and $w\in\mathrm{relint}(\Delta_n)$, we have
    \begin{equation}
        \nabla^2 g(w)[i][j]
        =\frac{1}{1-\alpha}\Tr\left[\frac{A_j^{\alpha}}{\Tr\left[A_j^{\alpha}Q_{\star}(w)^{1-\alpha}\right]}\left(L_1^{-1}+L_2\right)^{-1}\left(\frac{A_i^{\alpha}}{\Tr\left[A_i^{\alpha}Q_{\star}(w)^{1-\alpha}\right]}\right)\right]-1,\quad\forall i,j,
    \end{equation}
    where $L_1\in\mathcal{S}_{d,++}$ and $L_2\in\mathcal{S}_{d,+}$ are defined as follows:
    \begin{align}
        L_1(V)
        &\coloneqq \sum_{1 \leq k, l \leq d} h_{(1-\alpha)/\alpha}^{[1]}\left(\lambda_k^{\alpha}, \lambda_l^{\alpha}\right) u_k u_k^{*} V u_l u_l^{*}, \nonumber \\
        L_2(V)
        &\coloneqq \sum_{k=1}^{n} w[k] \left( \frac{ \Tr\left[{A_k^{\alpha}}{V}\right] A_k^{\alpha}}{ \Tr\left[{A_k^{\alpha}}{Q_{\star}(w)^{1-\alpha}}\right]^2} \right),\quad\forall V\in\mathcal{B}_{d},
    \end{align}
    where $Q_{\star}(w)=\sum_{k=1}^d\lambda_k u_k u_k^{*}$ is the eigendecomposition of $Q_{\star}(w)$, and $h_{(1-\alpha)/\alpha}^{[1]}$ is defined as in Lemma \ref{lemma:PowerInfDiff}.
\end{lemma}

\begin{proof}
    Let $w$ and $\alpha$ be defined as in Lemma \ref{lemma:HessianForm}.
    Let $Q_{\star}(w)=\sum_{k=1}^d\lambda_k u_k u_k^{*}\in\mathcal{B}_{d,++}$ denote the eigendecomposition of $Q_{\star}(w)$.
    By Lemma \ref{lemma:Danskin}, Lemma \ref{lemma:MeanDiff}, and the chain rule, we have, for any $i,j\in\Set{1,2,\dots,n}$,
    \begin{align}
        \label{eq:HessianFirstHalf}
        \nabla^2 g(w)[i][j]
        &=\frac{\partial}{\partial w[i]}\left(\nabla g(w)[j]\right)\nonumber\\
        &=\frac{\partial}{\partial w[i]}\frac{\log\left(\Tr\left[A_j^{\alpha}Q_{\star}(w)^{1-\alpha}\right]\right)}{1-\alpha}\nonumber\\
        &=\frac{1}{1-\alpha}\frac{\Tr\left[A_j^{\alpha}\left(\frac{\partial}{\partial w[i]}Q_{\star}(w)^{1-\alpha}\right)\right]}{\Tr\left[A_j^{\alpha}Q_{\star}(w)^{1-\alpha}\right]}.
    \end{align}
    On the other hand, by taking the partial derivative of both sides of the equation
    \begin{equation}
        Q_{\star}(w)^{1-\alpha}
        =\left(\sum_{s=1}^n\frac{\frac{w[s]}{\inner{w}{\mathbf{1}_n}}A_s^{\alpha}}{\Tr\left[A_s^{\alpha}Q_{\star}(w)^{1-\alpha}\right]}\right)^{(1-\alpha)/\alpha},
    \end{equation}
    which is given by Lemma \ref{eq:FixPtIsMin},
    we obtain the following identity:
    \begin{align}
        &\frac{\partial }{\partial w[i]}Q_{\star}(w)^{1-\alpha}\nonumber\\
        &\quad=\frac{\partial }{\partial w[i]}\left(\sum_{s=1}^{n}\frac{\frac{w[s]}{\inner{w}{\mathbf{1}_n}}A_s^{\alpha}}{\Tr\left[{A_s^{\alpha}}{Q_{\star}(w)^{1-\alpha}}\right]}\right)^{(1-\alpha)/\alpha}\nonumber\\
        &\quad=\left(\mathrm{D}_UU^{(1-\alpha)/\alpha}\right)\left(\frac{\partial }{\partial w[i]}\left(\sum_{s=1}^{n}\frac{\frac{w[s]}{\inner{w}{\mathbf{1}_n}}A_s^{\alpha}}{\Tr\left[{A_s^{\alpha}}{Q_{\star}(w)^{1-\alpha}}\right]}\right)\right)\nonumber\\
        &\quad=\sum_{1\leq k,l\leq d}h_{(1-\alpha)/\alpha}^{[1]}\left(\lambda_k^{\alpha},\lambda_l^{\alpha}\right)u_k u_k^{*}\left(\frac{\partial }{\partial w[i]}\left(\sum_{s=1}^{n}\frac{\frac{w[s]}{\inner{w}{\mathbf{1}_n}}A_s^{\alpha}}{\Tr\left[{A_s^{\alpha}}{Q_{\star}(w)^{1-\alpha}}\right]}\right)\right)u_l u_l^{*}\nonumber\\
        &\quad=L_1\left(\frac{\partial }{\partial w[i]}\left(\sum_{s=1}^{n}\frac{\frac{w[s]}{\inner{w}{\mathbf{1}_n}}A_s^{\alpha}}{\Tr\left[{A_s^{\alpha}}{Q_{\star}(w)^{1-\alpha}}\right]}\right)\right),\nonumber
    \end{align}
    where 
    \begin{equation}
        U
        =\sum_{s=1}^{n}\frac{\frac{w[s]}{\inner{w}{\mathbf{1}_n}}A_s^{\alpha}}{\Tr\left[{A_s^{\alpha}}{Q_{\star}(w)^{1-\alpha}}\right]}
        =Q_{\star}(w)^{\alpha},
    \end{equation}
    the second equality follows from the chain rule, and the third equality follows from Lemma \ref{lemma:PowerInfDiff}.
    The computation continues as follows:
    \begin{align}
        &\frac{\partial }{\partial w[i]}Q_{\star}(w)^{1-\alpha}\nonumber\\
        &\quad=L_1\left(\frac{\partial }{\partial w[i]}\left(\sum_{s=1}^{n}\frac{\frac{w[s]}{\inner{w}{\mathbf{1}_n}}A_s^{\alpha}}{\Tr\left[{A_s^{\alpha}}{Q_{\star}(w)^{1-\alpha}}\right]}\right)\right)\nonumber\\
        &\quad=L_1\left(\frac{\frac{1}{\inner{w}{\mathbf{1}_n}}A_i^{\alpha}}{\Tr\left[A_i^{\alpha}Q_{\star}(w)^{1-\alpha}\right]}\right)-L_1\left(\sum_{s=1}^{n}\frac{\frac{w[s]}{\inner{w}{\mathbf{1}_n}^2}A_s^{\alpha}}{\Tr\left[A_s^{\alpha}Q_{\star}(w)^{1-\alpha}\right]}\right)\nonumber\\
        &\quad\quad-L_1\left(\sum_{s=1}^{n}\frac{w[s]}{\inner{w}{\mathbf{1}_n}}\left(\frac{\Tr\left[{A_s^{\alpha}}\left({\frac{\partial }{\partial w[i]}Q_{\star}(w)^{1-\alpha}}\right)\right]A_s^{\alpha}}{\Tr\left[A_s^{\alpha}Q_{\star}(w)^{1-\alpha}\right]^2}\right)\right)\nonumber\\
        &\quad=L_1\left(\frac{A_i^{\alpha}}{\Tr\left[A_i^{\alpha}Q_{\star}(w)^{1-\alpha}\right]}\right)-L_1\left(Q_{\star}(w)^{\alpha}\right)-L_1\left(\sum_{s=1}^{n}w[s]\left(\frac{\Tr\left[{A_s^{\alpha}}\left({\frac{\partial }{\partial w[i]}Q_{\star}(w)^{1-\alpha}}\right)\right]A_s^{\alpha}}{\Tr\left[A_s^{\alpha}Q_{\star}(w)^{1-\alpha}\right]^2}\right)\right)\nonumber\\
        &\quad=L_1\left(\frac{A_i^{\alpha}}{\Tr\left[A_i^{\alpha}Q_{\star}(w)^{1-\alpha}\right]}\right)-L_1\left(Q_{\star}(w)^{\alpha}\right)-L_1\left(L_2\left(\frac{\partial }{\partial w[i]}Q_{\star}(w)^{1-\alpha}\right)\right),\nonumber
    \end{align}
    where the third equality follows from Lemma \ref{eq:FixPtIsMin} and the fact that $\inner{w}{\mathbf{1}_n}=1$.
    By rearrangement, we obtain
    \begin{equation}
        \frac{\partial }{\partial w[i]}Q_{\star}(w)^{1-\alpha}
        =\left(L_1^{-1}+L_2\right)^{-1}\left(\frac{A_i^{\alpha}}{\Tr\left[A_i^{\alpha}Q_{\star}(w)^{1-\alpha}\right]}\right)-\left(L_1^{-1}+L_2\right)^{-1}\left(Q_{\star}(w)^{\alpha}\right).
    \end{equation}
    To simplify the last term $\left(L_1^{-1}+L_2\right)^{-1}\left(Q_{\star}(w)^{\alpha}\right)$, we use Lemma \ref{eq:FixPtIsMin} and compute
    \begin{equation}
        L_1\left(Q_{\star}(w)^{\alpha}\right)
        =\frac{1-\alpha}{\alpha}Q_{\star}(w)^{1-\alpha},
    \end{equation}
    and
    \begin{equation}
        L_2\left(Q_{\star}(w)^{1-\alpha}\right)
        =Q_{\star}(w)^{\alpha}.
    \end{equation}
    It follows that
    \begin{equation}
        \left(L_1^{-1}+L_2\right)\left(Q_{\star}(w)^{1-\alpha}\right)
        =\left(\frac{\alpha}{1-\alpha}+1\right)Q_{\star}(w)^{\alpha}.
    \end{equation}
    Thus, we obtain
    \begin{equation}
        \label{eq:HessianSecondHalf}
        \frac{\partial }{\partial w[i]}Q_{\star}(w)^{1-\alpha}
        =\left(L_1^{-1}+L_2\right)^{-1}\left(\frac{A_i^{\alpha}}{\Tr\left[A_i^{\alpha}Q_{\star}(w)^{1-\alpha}\right]}\right)-(1-\alpha)Q_{\star}(w)^{1-\alpha}.
    \end{equation}
    The proof then follows by combining Equation \eqref{eq:HessianFirstHalf} and Equation \eqref{eq:HessianSecondHalf}.
\end{proof}

\begin{lemma}(Lieb-Ando Concavity \citep[Corollary 1.1]{Lieb1973})
\label{lemma:LiebAndo}
For any $W\in\mathcal{B}_d$, the mapping $(U,V)\mapsto \Tr\left[WU^{r}WV^{s}\right]$ is jointly concave in $(U,V)\in\mathcal{B}_{d,+}\times\mathcal{B}_{d,+}$ for all $r,s\in\mathbb{R}_{++}$ such that $r+s\leq 1$.
\end{lemma}
\begin{lemma}
    \label{lemma:ad-hoc}
    The function \( \psi(x) = (\exp(x) - \exp(-x))/{x} \) is increasing on \( \mathbb{R}_{+} \).
\end{lemma}
\begin{proof}
    We extend the domain of $\psi(x)$ by defining $\psi(0)=\lim_{x\to 0}\psi(x)=2$, so that $\psi$ is differentiable on $\mathbb{R}$.
    Then, it suffices to prove that the first-order derivative 
    \begin{align}
        \psi'(x) = \frac{x \exp(x) + x \exp(-x) - \exp(x) + \exp(-x)}{x^2}
    \end{align}
    is nonnegative. 
    Define
    \[
        \phi(x) = x \exp(x) + x \exp(-x) - \exp(x) + \exp(-x).
    \]
    Then, we have \( \phi(0) = 0 \), and
    \begin{align}
        \phi'(x) &= \exp(x) + x \exp(x) + \exp(-x) - x \exp(-x) - \exp(x) - \exp(-x) \nonumber \\
        &= x (\exp(x) - \exp(-x)) \nonumber\\
        &\geq 0, \quad \forall x\geq 0. \nonumber
    \end{align}
    Thus, \( \phi(x) \geq 0 \), and hence \( \psi'(x) = \phi(x)/x^2 \geq 0 \) for all \( x > 0 \), and $\psi'(0)=0$.
\end{proof}
\begin{lemma}
    \label{lemma:L1geqL2}
    Let $L_1,L_2$ be defined as in Lemma \ref{lemma:HessianForm}.
    Then, for any $\alpha\in\left[1/2,1\right)$, we have
    \begin{align}
        L_1^{-1}
        \geq\frac{\alpha}{1-\alpha}L_2.
    \end{align}
\end{lemma}
\begin{proof}
    Let $w,Q_{\star}(w),L_1$, and $L_2$ be defined as in Lemma \ref{lemma:HessianForm}, and  let $\alpha\in\left[1/2,1\right)$.
    Let $Q_{\star}(w)=\sum_{i=1}^d\lambda_i u_iu_i^{*}$ denote the eigendecomposition of $Q_{\star}(w)$.
    To prove that $L_1^{-1}\geq{\alpha L_2}/{(1-\alpha)}$, it suffices to prove that
    \begin{align}
        \Tr\left[UL_1^{-1}(U)\right]
        \geq \frac{\alpha}{1-\alpha}\Tr\left[UL_2(U)\right],\quad\forall U\in\mathcal{B}_d.
    \end{align}
    We begin by defining an auxiliary operator $L_3\in\mathcal{S}_{d,++}$ as
    \begin{equation}
        L_3(V)
        \coloneqq Q_{\star}(w)^{(1-\alpha)/2}VQ_{\star}(w)^{(1-\alpha)/2},\quad\forall V\in\mathcal{B}_d.
    \end{equation}
    Also, for each $j$, we define
    \begin{equation}
        \Lambda_j
        \coloneqq \frac{A_j^{\alpha}}{\Tr\left[A_j^{\alpha}Q_{\star}(w)^{1-\alpha}\right]}.
    \end{equation}
    The proof then proceeds by computing an upper bound on $L_2$:
    \begin{align}
        \Tr\left[UL_2(U)\right]
        &=\sum_{j=1}^n w[j]\Tr\left[\Lambda_j U\right]^2\nonumber\\
        &=\sum_{j=1}^n w[j]\Tr\left[\left(L_3(\Lambda_j)\right)^{1/2}\left(\left(L_3(\Lambda_j)\right)^{1/4}L_3^{-1}(U)\left(L_3(\Lambda_j)\right)^{1/4}\right)\right]^2\nonumber\\
        &\leq \sum_{j=1}^n w[j]\Tr\left[L_3(\Lambda_j)\right]\Tr\left[\left(\left(L_3(\Lambda_j)\right)^{1/4}L_3^{-1}(U)\left(L_3(\Lambda_j)\right)^{1/4}\right)^2\right]\nonumber\\
        &=\sum_{j=1}^n w[j]\Tr\left[\left(\left(L_3(\Lambda_j)\right)^{1/4}L_3^{-1}(U)\left(L_3(\Lambda_j)\right)^{1/4}\right)^2\right],\quad\forall U\in\mathcal{B}_d,
    \end{align}
    where the inequality follows from the H\"{o}lder inequality (Lemma \ref{ineq:HolderEq}).
    The computation of the upper bound continues as follows:
    \begin{align}
        \Tr\left[UL_2(U)\right]
        &\leq\sum_{j=1}^n w[j]\Tr\left[\left(\left(L_3(\Lambda_j)\right)^{1/4}L_3^{-1}(U)\left(L_3(\Lambda_j)\right)^{1/4}\right)^2\right]\nonumber\\
        &=\sum_{j=1}^n w[j]\Tr\left[L_3^{-1}(U)\left(L_3(\Lambda_j)\right)^{1/2}L_3^{-1}(U)\left(L_3(\Lambda_j)\right)^{1/2}\right]\nonumber\\
        &\leq \Tr\left[L_3^{-1}(U)\left(\sum_{k=1}^n w[k]L_3(\Lambda_k)\right)^{1/2}L_3^{-1}(U)\left(\sum_{l=1}^n w[l]L_3(\Lambda_l)\right)^{1/2}\right]\nonumber\\
        &= \Tr\left[L_3^{-1}(U)\left(L_3\left(\sum_{k=1}^n w[k]\Lambda_k\right)\right)^{1/2}L_3^{-1}(U)\left(L_3\left(\sum_{l=1}^n w[l]\Lambda_l\right)\right)^{1/2}\right]\nonumber\\
        &= \Tr\left[L_3^{-1}(U)Q_{\star}(w)^{1/2}L_3^{-1}(U)Q_{\star}(w)^{1/2}\right]\nonumber\\
        &=\Tr\left[Q_{\star}(w)^{\alpha-1/2}UQ_{\star}(w)^{\alpha-1/2}U\right],\quad\forall U\in\mathcal{B}_d,
    \end{align}
    where the second inequality follows from the Lieb-Ando concavity (Lemma \ref{lemma:LiebAndo}) and the Jensen inequality, and the third equality follows from Lemma \ref{eq:FixPtIsMin}.
    To facilitate comparison with $L_1^{-1}$, we further express the upper bound as
    \begin{align}
        \label{ineq:L2Split}
        \Tr\left[UL_2(U)\right]
        &\leq\Tr\left[Q_{\star}(w)^{\alpha-1/2}UQ_{\star}(w)^{\alpha-1/2}U\right]\nonumber\\
        &=\sum_{\substack{1\leq k,l\leq d,\\\lambda_k=\lambda_l}}\lambda_k^{2\alpha-1}\abs{u_k^{*}Uu_l}^2+\sum_{\substack{1\leq k,l\leq d,\\\lambda_k\neq\lambda_l}}\lambda_k^{\alpha-1/2}\lambda_l^{\alpha-1/2}\abs{u_k^{*}Uu_l}^2,\quad\forall U\in\mathcal{B}_d.
    \end{align}
    On the other hand, for $L_1^{-1}$, we write
    \begin{align}
        \label{eq:L_1Split}
        \Tr\left[UL_1^{-1}(U)\right]
        &=\sum_{\substack{1\leq k,l\leq d,\\\lambda_k=\lambda_l}}\frac{\alpha}{1-\alpha}\lambda_k^{2\alpha-1}\abs{u_k^{*}Uu_l}^2+\sum_{\substack{1\leq k,l\leq d,\\\lambda_k\neq\lambda_l}}\frac{\lambda_k^{\alpha}-\lambda_l^{\alpha}}{\lambda_k^{1-\alpha}-\lambda_l^{1-\alpha}}\abs{u_k^{*}Uu_l}^2.
    \end{align}
    Based on \eqref{ineq:L2Split} and \eqref{eq:L_1Split}, it remains to prove that
    \begin{align}
        \frac{\alpha}{1-\alpha}\lambda_k^{\alpha-1/2}\lambda_l^{\alpha-1/2}
        \leq\frac{\lambda_k^{\alpha}-\lambda_l^{\alpha}}{\lambda_k^{1-\alpha}-\lambda_l^{1-\alpha}},\quad\forall \lambda_k\neq\lambda_l,
    \end{align}
    which is equivalent to proving that
    \begin{align}
        \label{ineq:ToBeShowByAdHoc}
        \frac{\lambda_k^{1/2}\lambda_l^{\alpha-1/2}-\lambda_k^{\alpha-1/2}\lambda_l^{1/2}}{1-\alpha}
        \leq\frac{\lambda_k^{\alpha}-\lambda_l^{\alpha}}{\alpha},\quad\forall \lambda_k>\lambda_l.
    \end{align}
    To prove \eqref{ineq:ToBeShowByAdHoc}, we define $r=\lambda_k/\lambda_l>1$, and rewrite the inequality as
    \begin{align}
        \frac{r^{(1-\alpha)/2}-r^{(\alpha-1)/2}}{1-\alpha}\leq\frac{r^{\alpha/2}-r^{-\alpha/2}}{\alpha},\quad\forall r>1,
    \end{align}
    which is in turn equivalent to the following inequality:
    \begin{align}
        \frac{r^{(1-\alpha)/2}-r^{(\alpha-1)/2}}{(1-\alpha)\log(r)/2}\leq\frac{r^{\alpha/2}-r^{-\alpha/2}}{\alpha\log(r)/2},\quad\forall r>1,
    \end{align}
    whose validity follows from Lemma \ref{lemma:ad-hoc} and the fact that $\alpha\in\left[1/2,1\right)$.
    This concludes the proof.
\end{proof}

\begin{proof}(Lemma \ref{lemma:HessianBd})
    Let $\alpha\in\left[1/2,1\right)$ and $w\in\mathrm{relint}\left(\Delta_n\right)$.
    Let $L_1,L_2$, and $Q_{\star}(w)=\sum_{k=1}^d\lambda_ku_ku_k^{*}$ be defined as in Lemma \ref{lemma:HessianForm}.
    Also, for each $j$, define
    \begin{equation}
        \Lambda_j
        \coloneqq \frac{A_j^{\alpha}}{\Tr\left[A_j^{\alpha}Q_{\star}(w)^{1-\alpha}\right]}.
    \end{equation}
    We aim to prove Lemma \ref{lemma:HessianBd} by showing that $\mathrm{Diag}\left(\sqrt{w}\right)\nabla^2 g(w)\mathrm{Diag}\left(\sqrt{w}\right)\leq I_n$.
    For any $v\in\mathbb{R}^n$, by Lemma \ref{lemma:HessianForm},
    we have
    \begin{align}
        \label{ineq:SqrtwHessSqrtw}
        &\inner{\mathrm{Diag}\left(\sqrt{w}\right)\nabla^2 g(w)\mathrm{Diag}\left(\sqrt{w}\right)v}{v}\nonumber\\
        &\quad=\Tr\left[\left(\sum_{j=1}^n\sqrt{w[j]}v[j]\Lambda_j\right)\left((1-\alpha)L_1^{-1}+(1-\alpha)L_2\right)^{-1}\left(\sum_{i=1}^n\sqrt{w[i]}v[i]\Lambda_i\right)\right]-\inner{v}{\sqrt{w}}^2\nonumber\\
        &\quad\leq\Tr\left[\left(\sum_{j=1}^n\sqrt{w[j]}v[j]\Lambda_j\right)\left((1-\alpha)L_1^{-1}+(1-\alpha)L_2\right)^{-1}\left(\sum_{i=1}^n\sqrt{w[i]}v[i]\Lambda_i\right)\right].
    \end{align}
    On the other hand, we define
    \begin{equation}
        L_4
        \coloneqq (1-\alpha)L_1^{-1}-\alpha L_2,
    \end{equation}
    which is positive semi-definite by Lemma \ref{lemma:L1geqL2}.
    To simplify \eqref{ineq:SqrtwHessSqrtw}, we define an auxiliary matrix $G\in\mathcal{B}_{n,+}$ by
    \begin{equation}
        G[i][j]
        \coloneqq \Tr\left[\left(\sqrt{w[j]}\Lambda_j\right)\left(L_4+L_2\right)^{-1}\left(\sqrt{w[i]}\Lambda_i\right)\right],\quad\forall i,j,
    \end{equation}
    where $(L_4+L_2)^{-1}$ is well-defined since $L_4+L_2\geq (1-\alpha)L_1^{-1}>0$.
    Then, it follows that
    \begin{align}
        &\inner{\mathrm{Diag}\left(\sqrt{w}\right)\nabla^2 g(w)\mathrm{Diag}\left(\sqrt{w}\right)v}{v}\nonumber\\
        &\quad\leq\Tr\left[\left(\sum_{j=1}^n\sqrt{w[j]}v[j]\Lambda_j\right)\left((1-\alpha)L_1^{-1}+(1-\alpha)L_2\right)^{-1}\left(\sum_{i=1}^n\sqrt{w[i]}v[i]\Lambda_i\right)\right]\nonumber\\
        &\quad=\Tr\left[\left(\sum_{j=1}^n\sqrt{w[j]}v[j]\Lambda_j\right)\left(L_4+L_2\right)^{-1}\left(\sum_{i=1}^n\sqrt{w[i]}v[i]\Lambda_i\right)\right]\nonumber\\
        &\quad=\inner{Gv}{v}.\nonumber
    \end{align}
    Therefore, to prove that $\mathrm{Diag}\left(\sqrt{w}\right)\nabla^2 g(w)\mathrm{Diag}\left(\sqrt{w}\right)\leq I_n$, it suffices to show that $G\leq I_n$, which we will prove by establishing that $G^2\leq G$.
    For each $i,j$, we compute
    \begin{align}
        G^2[i][j]
        &=\sum_{k=1}^n G[i][k]G[k][j]\nonumber\\
        &=\sum_{k=1}^n\Tr\left[\left(\sqrt{w[k]}\Lambda_k\right)\left(L_4+L_2\right)^{-1}\left(\sqrt{w[i]}\Lambda_i\right)\right]\Tr\left[\left(\sqrt{w[k]}\Lambda_k\right)\left(L_4+L_2\right)^{-1}\left(\sqrt{w[j]}\Lambda_j\right)\right]\nonumber\\
        &=\sum_{k=1}^nw[k]\Tr\left[\Lambda_k\left(L_4+L_2\right)^{-1}\left(\sqrt{w[i]}\Lambda_i\right)\right]\Tr\left[\Lambda_k\left(L_4+L_2\right)^{-1}\left(\sqrt{w[j]}\Lambda_j\right)\right]\nonumber\\
        &=\Tr\left[\left(L_4+L_2\right)^{-1}\left(\sqrt{w[i]}\Lambda_i\right) L_2\left(\left(L_4+L_2\right)^{-1}\left(\sqrt{w[j]}\Lambda_j\right)\right)\right]\nonumber\\
        &=\Tr\left[\left(\sqrt{w[i]}\Lambda_i\right) \left(L_4+L_2\right)^{-1}\left(\sqrt{w[j]}\Lambda_j\right)\right]\nonumber\\
        &\quad-\Tr\left[\left(L_4+L_2\right)^{-1}\left(\sqrt{w[i]}\Lambda_i\right) L_4\left(\left(L_4+L_2\right)^{-1}\left(\sqrt{w[j]}\Lambda_j\right)\right)\right]\nonumber\\
        &=G[i][j]-\Tr\left[\left(L_4+L_2\right)^{-1}\left(\sqrt{w[i]}\Lambda_i\right) L_4\left(\left(L_4+L_2\right)^{-1}\left(\sqrt{w[j]}\Lambda_j\right)\right)\right].\nonumber
    \end{align}
    It follows that
    \begin{align}
        &\inner{G^2 v}{v}\nonumber\\
        &\quad= \inner{Gv}{v}-\Tr\left[\left(\sum_{i=1}^n\left(L_4+L_2\right)^{-1}\left(\sqrt{w[i]}v[i]\Lambda_i\right) \right)L_4\left(\sum_{j=1}^n\left(L_4+L_2\right)^{-1}\left(\sqrt{w[j]}v[j]\Lambda_j\right) \right)\right]\nonumber\\
        &\quad\leq \inner{Gv}{v},\quad\forall v\in\mathbb{R}^n.
    \end{align}
    This implies that $G^2\leq G$.
    Moreover, since $\nabla^2 g(w)\geq0$ and $\mathrm{Diag}\left(\sqrt{w}\right)\nabla^2 g(w)\mathrm{Diag}\left(\sqrt{w}\right)\leq G$, we conclude that $0\leq G^2\leq G\leq I_n$.
    Thus, we have
    \begin{align}
        \mathrm{Diag}\left(\sqrt{w}\right)\nabla^2 g(w)\mathrm{Diag}\left(\sqrt{w}\right)\leq I_n.
    \end{align}
    The proof then follows from
    \begin{align}
        \nabla^2 g(w)
        \leq \mathrm{Diag}\left(\sqrt{w}\right)^{-1}I_n\mathrm{Diag}\left(\sqrt{w}\right)^{-1}
        =\mathrm{Diag}\left(w\right)^{-1}.
    \end{align}
\end{proof}

\section{Equivalence between the Augustin Iteration \citep{Augustin1978} and the Algorithm of \citet{Cheung2018}}
\label{appendix:EqAugTatonnement}
This section is devoted to showing the equivalence between the Augustin iteration \citep{Augustin1978} and the algorithm proposed by \citet{Cheung2018} when applied to computing the equilibrium prices in Fisher markets introduced in Section \ref{subsec:FisherDef}, assuming all $\rho_j = \rho$.

The algorithm proposed by \citet{Cheung2018} iterates as
\begin{equation}
    \label{eq:CheungIter}
    p_{t+1}
    =p_t\odot x(p_t),
\end{equation}
where $p_t\in\mathbb{R}_{++}^d$ is the $t^{\text{th}}$ iterate, and
\begin{equation}
    x(p)
    \coloneqq \sum_{j=1}^n w[j]\frac{a_j^{1/(1-\rho)}\odot p^{-1/(1-\rho)}}{\inner{a_j^{1/(1-\rho)}}{p^{-\rho/(1-\rho)}}}
\end{equation}
is the total demand vector as defined in \eqref{def:FisherTotalDemand}.
On the other hand, the Augustin iteration for solving the optimization problem~\eqref{def:ClassicalAugMean} iterates as
\begin{equation}
    q_{t+1}
    =q_t\odot\left(-\nabla f(q_t)\right)
    =q_t\odot\left(\sum_{j=1}^{n}w[j]\frac{a_j^{\alpha}\odot q_t^{-\alpha}}{\inner{a_j^{\alpha}}{q_t^{1-\alpha}}}\right),
\end{equation}
which exactly recovers the iteration rule \eqref{eq:CheungIter} by setting $q_t = p_t$, $\alpha = 1/(1-\rho)$, and letting all $a_j$ be the same as in the Fisher market setting.

\section{Comparison with the Error Measure of \citet{Cole2008}}
This section is devoted to showing, via Lemma \ref{lemma:MetricComparable}, that the error measure adopted by \citet{Cole2008} is comparable to the Thompson metric.

We use Lemma \ref{lemma:Nesterov-adhoc} to prove Lemma \ref{lemma:MetricComparable}.
\begin{lemma}(\citep[Lemma 5.1.5]{Nesterov2018a})
    \label{lemma:Nesterov-adhoc}
    For any $r>1$, we have
    \begin{align}
        \frac{2(r-1)}{r+1}
        \leq \log(r)
        \leq \frac{(r-1)(r+1)}{2r}.\nonumber
    \end{align}
    On the other hand, for any $r\in(0,1)$, we have
    \begin{align}
        \frac{2(1-r)}{r+1}
        \leq \log\left(\frac{1}{r}\right)
        \leq \frac{(1-r)(1+r)}{2r}.\nonumber
    \end{align}
\end{lemma}
\begin{lemma}
    \label{lemma:MetricComparable}
    For any $u,v\in\mathbb{R}_{++}^d$ such that $d_{\mathrm{T}}(v,u)<\log(3)$, we have
    \begin{equation}
        \frac{1}{3}\max_{i}\abs{1-\frac{u[i]}{v[i]}}
        \leq d_{\mathrm{T}}(v,u)
        \leq 3\max_{i}\abs{1-\frac{u[i]}{v[i]}}.
    \end{equation}
\end{lemma}
\begin{proof}
    Let $u,v$ be as defined in Lemma \ref{lemma:MetricComparable}.
    By Definition \ref{eq:Thompson}, we have
    \begin{equation}
        d_{\mathrm{T}}(v,u)
        =\max_{i}\abs{\log\left(\frac{u[i]}{v[i]}\right)}.
    \end{equation}
    Therefore, it suffices to prove that
    \begin{align}
        \frac{1}{3}(r-1)
        \leq \log(r)
        \leq 3(r-1),\quad\forall r\in(1,3),
    \end{align}
    and
    \begin{align}
        \frac{1}{3}(1-r)
        \leq \log\left(\frac{1}{r}\right)
        \leq 3(1-r),\quad\forall r\in(1/3,1),
    \end{align}
    which follow directly from Lemma \ref{lemma:Nesterov-adhoc}.
\end{proof}

\section{Counterexample to the Convergence Guarantee of \citet{Dvijotham2022}}
\label{appendix:DvijothamCounterexample}
This appendix provides a counterexample to the convergence guarantee of the best response dynamic of \citet{Dvijotham2022}.
After presenting the counterexample, we point out the incorrect part of the proof provided by \citet{Dvijotham2022}.

At each $t^{\text{th}}$ iteration, each $i^{\text{th}}$ seller assumes that the prices set by the other sellers remain fixed.
The seller then updates their price so that the demand for their good matches its supply.
This update rule is called the best response dynamic.
In particular, define a mapping $F_{\mathrm{br}}:\mathbb{R}^d\times \Set{1,2,\dots,d}\to\mathbb{R}$ such that for all $i$, we have $x(\tilde{p}_i)[i]=1$, where $\tilde{p}_i[k]=p[k]$ for all $k\neq i$ and $\tilde{p}_i[i]=F_{\mathrm{br}}(p,i)$.
Then, in the best response dynamic, each price $p_{t+1}[i]$ is set to $F_{\mathrm{br}}(p_t,i)$ for all $i$ at each iteration $t$.
\citet[Corollary 18]{Dvijotham2022} claimed that the distance, measured in the Thompson metric, between the iterate and the equilibrium of a Fisher market with CES utilities sharing a common CES parameter $\rho\in(0,1)$ satisfies the following contraction property:
\begin{align}
    d_{\mathrm{T}}(p_{\star},p_{t+1})
    \leq \sigma^{t} d_{\mathrm{T}}(p_{\star},p_{1}),
\end{align}
where the contraction factor is given by
\begin{align}
    \sigma
    =1-\frac{1}{1+\frac{\rho}{1-\rho}\left(1-\min_{i}\max_{j}\frac{a_j[i]^{1/\rho}}{a_j[i]^{1/\rho}+\sum_{k\neq i}a_{j}[k]^{1/\rho}\left(\frac{p_{\mathrm{min}}}{p_{\mathrm{max}}}\right)^{\rho/(1-\rho)}}\right)}.
\end{align}
Here,
\begin{align}
    p_{\mathrm{min}}
    =\min\Set{1,\min_{i}\sum_{j=1}^n\frac{w[j]a_j[i]^{1/\rho}}{\inner{\mathbf{1}_d}{a_j^{1/\rho}}}},
\end{align}
and
\begin{align}
    p_{\mathrm{max}}
    =\max\Set{1,\max_{i}\sum_{j=1}^n\frac{w[j]a_j[i]^{1/\rho}}{\inner{\mathbf{1}_d}{a_j^{1/\rho}}}}.
\end{align}

To construct a counterexample, we consider a market with a single buyer and two sellers.
Let $a_1=(0.5,0.5)$ and $\rho=0.5$.
Then, we have $p_{\mathrm{min}}=0.5$, $p_{\mathrm{max}}=1$, and $\sigma=0.25$.
In this case, the equilibrium price vector is $p_{\star}=(0.5,0.5)$.

On the other hand, let the initial price vector be $p_1=(0.9,0.9)$.
Following the best response dynamic, the next price vector is given by
\begin{align}
    p_2
    =(0.6,0.6).
\end{align}
We then have
\begin{align}
    d_{\mathrm{T}}(p_{\star},p_2)
    =\log(1.2)
    \approx 0.1823
    >\sigma d_{\mathrm{T}}(p_{\star},p_1)
    =0.25\log(1.8)
    \approx 0.1469,
\end{align}
which violates the contraction property claimed by \citet{Dvijotham2022}.

To prove such a contraction property, \citet{Dvijotham2022} implicitly use Lemma 16 in their paper.
The proof of Lemma~16 in \citet{Dvijotham2022}
contains the following incorrect inequality:
\begin{align}
    \label{ineq:Incorrect}
    \frac{x(\hat{p}_i)[i]}{x(\tilde{p}_i)[i]}
    \leq \left(\frac{F_{\mathrm{br}}(\lambda p,i)}{\lambda F_{\mathrm{br}}(p,i)}\right)^{1/(1-\sigma)},
\end{align}
for all $\lambda\in(0,1)$.
Here, $\hat{p}_i[k]=\tilde{p}_i[k]=\lambda p[k]$ for all $k\neq i$, $\hat{p}_i[i]=\lambda F_{\mathrm{br}}( p,i)$, and $\tilde{p}_i[i]=F_{\mathrm{br}}(\lambda p,i)$.
For a concrete counterexample, set $\lambda=49/54$ and take $p=(0.9,0.9)$ in the example constructed above.
Then, we have 
$F_{\mathrm{br}}(\lambda p,1)=7/12$,
$\lambda F_{\mathrm{br}}(p,1)=49/90$,
$\hat{p}_1=(49/90,49/60)$, and $\tilde{p}_1=(7/12,49/60)$.
This implies that $x(\hat{p}_1)[1]=54/49$ and $x(\tilde{p}_1)[1]=1$.
On the other hand, the right-hand side of inequality \eqref{ineq:Incorrect} is given by $(15/14)^{4/3}\approx 1.0964<54/49\approx 1.102$, which contradicts the inequality itself.

\end{document}